\documentclass[a4paper,UKenglish]{lipics-v2018}

\usepackage{microtype}%if unwanted, comment out or use option "draft"
\nolinenumbers %UNCOMMENT FOR ARXIV SUBMITION

\title{A Waste-Efficient Algorithm for Single-Droplet Sample Preparation on Microfluidic Chips}

\titlerunning{Mixing Graphs}

\author{Miguel Coviello Gonzalez}{Department of Computer Science\\ University of California at Riverside}{}{}{}

\author{Marek Chrobak$\star$}{Department of Computer Science\\ University of California at Riverside}{}{}{}

\Copyright{The copyright is retained by the authors}

\authorrunning{M. Coviello Gonzalez and M. Chrobak}

\subjclass{%
}


\funding{Research supported by NSF grant CCF-1536026.}

\newcommand{\mareksmargincomment}[1]%
    {{%
      \marginpar{{\tiny\begin{minipage}{0.5in}
                       \begin{flushleft}
                          {\color{red}MCh} {#1}
                       \end{flushleft}
                       \end{minipage}
                }}
    }}

\newcommand{\etal}{{\emph{et~al.}}}

\newcommand{\braced}[1]{{ \left\{ #1 \right\} }}

\newcommand{\ceiling}[1]{\lceil#1\rceil}

\newcommand{\calB}{{\cal B}}

\newcommand{\tildei}{{\tilde{\imath}}}
\newcommand{\tildej}{{\tilde{\jmath}}}

\newcommand{\tildephi}{{\tilde{\phi}}} 
 \newcommand{\tildepsi}{{\tilde{\psi}}}

\newcommand{\myparagraph}[1]{{\medskip\noindent\textbf{#1}.}}
\newcommand{\emparagraph}[1]{{\smallskip\noindent\textit{#1}.}} 
\newcommand{\mycase}[1]{{\underline{Case~#1}:}}

\newcommand{\integers}{{\mathbb Z}}

\newcommand{\nonnegintegers}{{\mathbb{Z}_{\ge 0}}}

\newcommand{\binrep}{{\textsf{bin}}}

\newcommand{\nofdroplets}[1]{{|#1|}}

\newcommand{\precision}{{\textsf{prec}}}

\newcommand{\initJ}{J_{{\scriptscriptstyle\textrm{init}}}}

\newcommand{\MinMix}{\texttt{Min-Mix}}
\newcommand{\DMRW}{\texttt{DMRW}}
\newcommand{\REMIA}{\texttt{REMIA}}
\newcommand{\GORMA}{\texttt{GORMA}}
\newcommand{\WARA}{\texttt{WARA}}

\newcommand{\ILP}{\texttt{ILP}}
\newcommand{\RPRIS}{\texttt{RPRIS}}
\newcommand{\MATLAB}{\texttt{MATLAB}}

\newcommand{\half}{{\textstyle\frac{1}{2}}}
\newcommand{\onehalf}{{\textstyle\frac{1}{2}}}

\newcommand{\onefourth}{{\textstyle\frac{1}{4}}}
\newcommand{\threefourths}{{\textstyle\frac{3}{4}}}

\newcommand{\oneeighth}{{\textstyle\frac{1}{8}}}
\newcommand{\threeeighths}{{\textstyle\frac{3}{8}}}
\newcommand{\fiveeighths}{{\textstyle\frac{5}{8}}}
\newcommand{\seveneighths}{{\textstyle\frac{7}{8}}}

\newcommand{\onesixteenth}{{\textstyle\frac{1}{16}}}

\newcommand{\fivesixteenths}{{\textstyle\frac{5}{16}}}

\newcommand{\ninesixteenths}{{\textstyle\frac{9}{16}}}
\newcommand{\elevensixteenths}{{\textstyle\frac{11}{16}}}
\newcommand{\thirteensixteenths}{{\textstyle\frac{13}{16}}}

\newtheorem{observation}{Observation}
\newtheorem{claim}{Claim}

\begin{document}
	
\maketitle

\begin{abstract} 
	We address the problem of designing micro-fluidic chips for
	sample preparation, which is a crucial step in many experimental processes in chemical and biological sciences.
	One of the objectives of sample preparation is to dilute the sample fluid, called reactant, 
	using another fluid called buffer, to produce desired volumes of
	fluid with prespecified reactant concentrations. In the model we adopt,
	these fluids are manipulated in discrete volumes called droplets.
	The dilution process is represented by a \emph{mixing graph}
	whose nodes represent 1-1 micro-mixers and edges represent channels for transporting fluids.
	In this work we focus on designing such mixing graphs when the given sample (also referred to 
	as the \emph{target}) consists of a single-droplet, and the objective is to minimize total fluid waste.
	Our main contribution is an efficient algorithm called $\RPRIS$ that
	guarantees a better provable worst-case bound on waste and 
	significantly outperforms state-of-the-art algorithms in experimental comparison. 
\end{abstract}

%%%%%%%%%%%%%%%%%%%%%%%%%%%%%%%%%%%%%%%%%%%%%%%%%%%%%%%%%%%%%%%%%%%%%%%%%%%%%%
%%%%%%%%%%%%%%%%%%%%%%%%%%%%%%%%%%%%%%%%%%%%%%%%%%%%%%%%%%%%%%%%%%%%%%%%%%%%%% 

\section{Introduction}
\label{sec: introduction}
%%%%%%%%%%%%%%%%%%%%%%%%%%%%%%%%%%%%%%%%%%%%%%%%%%%%%%%%%%%%%%%%%%%%%%%%%%%%%%
%%%%%%%%%%%%%%%%%%%%%%%%%%%%%%%%%%%%%%%%%%%%%%%%%%%%%%%%%%%%%%%%%%%%%%%%%%%%%% 

%\section{Introduction}
%\label{sec: introduction}
%\input{1_introduction.tex}

%%%%%%%%%%%%%%%%%%%%%%%%%%%%%%%%%%%%%%%%%%%%%%%%%%%%%%%%%%%%%%%%%%%%%%%%%%%%%%
%%%%%%%%%%%%%%%%%%%%%%%%%%%%%%%%%%%%%%%%%%%%%%%%%%%%%%%%%%%%%%%%%%%%%%%%%%%%%%

Microfluidic chips are miniature devices that can manipulate tiny amounts of fluids on a small chip 
and can perform, automatically, various laboratory functions such as dispensing, mixing, filtering and detection. 
They play an increasingly important role in today's science and
technology, with applications in environmental or medical monitoring, protein or DNA analysis, 
drug discovery, physiological sample analysis, and cancer research.

These chips often contain modules whose function is to mix fluids. One application where 
fluid mixing plays a  crucial role is sample preparation for some biological or chemical experiments.
When preparing such samples, one of the objectives is to produce desired volumes of the fluid of interest, 
called \emph{reactant}, diluted to some specified concentrations by
mixing it with another fluid called \emph{buffer}. As an example, an experimental study may require a 
sample that consists of $6 \mu L$ of reactant with concentration
$10\%$,  $9 \mu L$ of reactant with concentration $20\%$, and $3 \mu L$ of reactant with 
concentration $40\%$. Such multiple-concentration samples are often required in
toxicology or pharmaceutical studies, among other applications.

There are different models for fluid mixing in the literature and multiple 
technologies for manufacturing fluid-mixing microfluidic chips.  
(See the survey in~\cite{bhattacharya2014algorithmic} or the recent book~\cite{bhattacharya_book_2019algorithms}
for more information on different models and algorithmic issues related to fluid mixing.)
In this work we assume the \emph{droplet-based} model, where the fluids are manipulated
in discrete quantities called \emph{droplets}. For convenience, we will
identify droplets by their reactant concentrations, 
which are numbers in the interval $[0,1]$ with finite binary precision. In particular, a
droplet of reactant is denoted by $1$ and a droplet of buffer by $0$.
We focus on the mixing technology that utilizes modules
called \emph{1-1 micro-mixers}. A micro-mixer has two inlets and two outlets. 
It receives two droplets of fluid, one
in each inlet, mixes these droplets perfectly, and produces two droplets
of the mixed fluid, one on each outlet. (Thus,
if the inlet droplets have reactant concentrations $a$ and $b$,
then the two outlet droplets each will have concentration $\half(a+b)$.)
Input droplets are injected into the chip via droplet dispensers and output droplets are
collected in droplet collectors. All these components are
connected via micro-channels that transport droplets, forming naturally an acyclic graph that we call
a \emph{mixing graph}, whose source nodes are fluid dispensers, 
internal nodes (of in-degree and out-degree $2$)
are micro-mixers, and sink nodes are droplet collectors.
Graph $G_1$ in Figure~\ref{fig: mixing graph example} illustrates an example of a mixing graph.

\begin{figure}[ht]
	\begin{center}
		\includegraphics[width = 3.5in]{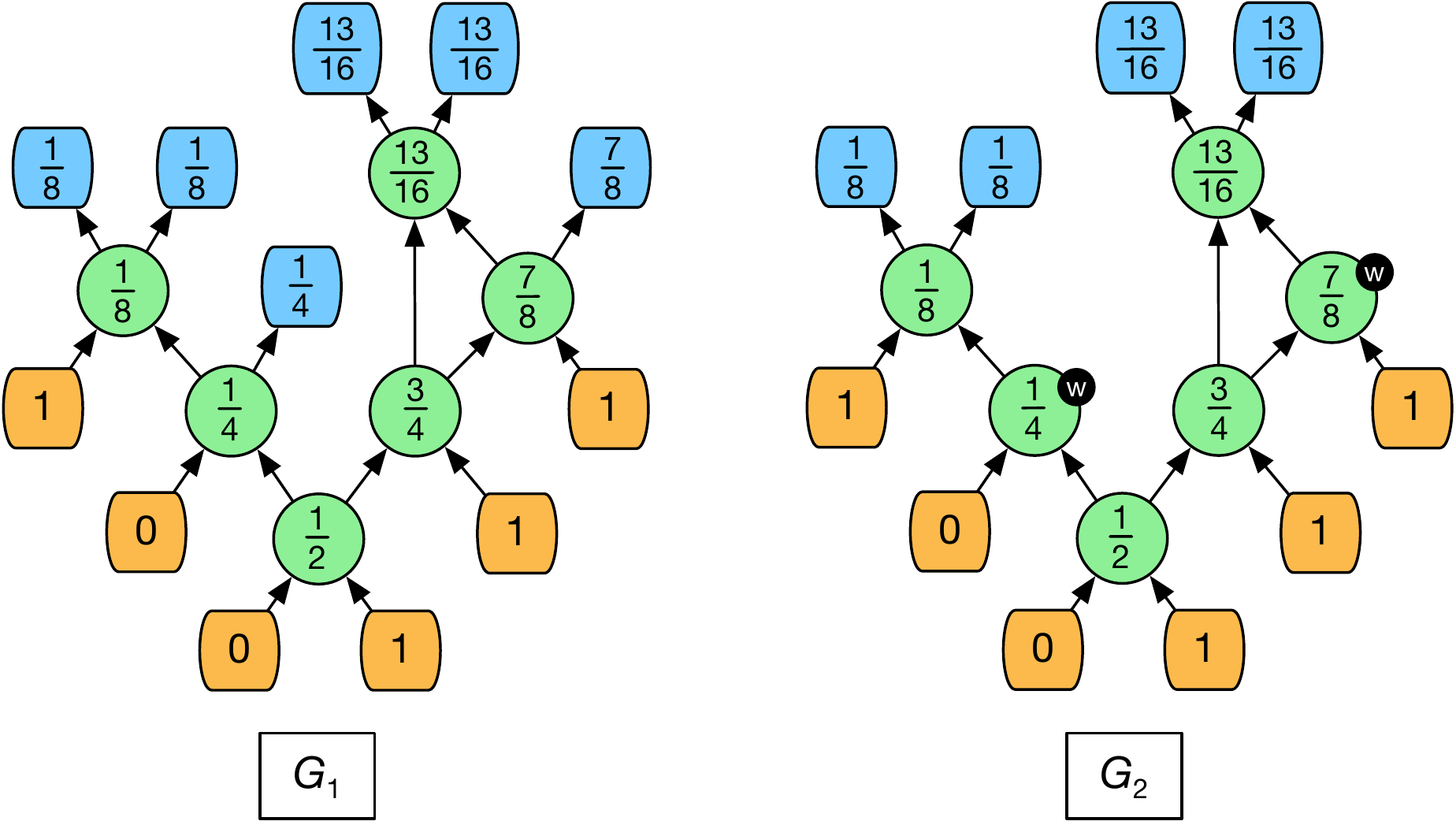}
		\caption{On the left, a mixing graph $G_1$ that produces droplet set
			$\braced{\oneeighth,\oneeighth,\onefourth,\thirteensixteenths,\thirteensixteenths,\seveneighths}$
			from input set $I = \braced{0,0,1,1,1,1}$.    
			Numbers on the micro-mixers (internal nodes) represent droplet concentrations produced by these micro-mixers. 
			If only some of the produced droplets are needed, the remaining droplets are
			designated as waste. 
			This is illustrated by the mixing graph $G_2$ that produces droplets
			$\braced{\oneeighth,\oneeighth,\thirteensixteenths,\thirteensixteenths}$.
			Small black circles labeled ``w'' on micro-mixers represent droplets of waste.}
		\label{fig: mixing graph example}
	\end{center}
\end{figure}

Given some target set of droplets with specified reactant concentrations, the objective is to design
a mixing graph that produces these droplets from pure reactant and buffer droplets, while
optimizing some objective function. Some target sets can be produced only if we allow
the mixing graph to also produce some superfluous amount of fluid that we refer to as
\emph{waste}; see graph $G_2$ in Figure~\ref{fig: mixing graph example}. 
One natural objective function is to minimize the number of waste droplets
(or equivalently, the total number of input droplets). As reactant is typically more
expensive than buffer, one other common objective is to minimize the reactant usage. 
Yet another possibility is to minimize the number of micro-mixers 
or the depth of the mixing graph. 
There is growing literature on developing techniques and algorithms for
designing such mixing graphs that attempt to optimize some of the above criteria. 

%%%%%%%%%%%%%%%%%%

\myparagraph{State-of-the-art} 
Most of the earlier papers on this topic studied designing mixing graphs for single-droplet targets.
This line of research was pioneered by Thies~{\etal}~\cite{thies2008abstraction}, 
who proposed an algorithm  called $\MinMix$ that constructs a mixing graph for a single-droplet target
with the minimum number of mixing operations.
Roy~{\etal}~\cite{roy2010optimization} developed an algorithm called $\DMRW$ designed to minimize waste.
Huang~{\etal}~\cite{huang2012reactant} considered minimizing
reactant usage, and proposed an algorithm called~{\REMIA}. Another algorithm called~{\GORMA}, 
for minimizing reactant usage and based on a branch-and-bound technique, 
was developed by Chiang~{\etal}~\cite{chiang2013graph}.

The algorithms listed above are heuristics, with no formal performance guarantees.  
An interesting attempt to develop an algorithm that minimizes waste, for target sets with multiple droplets, was 
reported by Dinh~{\etal}~\cite{dinh2014network}.
Their algorithm, that we refer to as $\ILP$, is based on a reduction to integer linear programming and, 
since their integer program could be exponential in the precision $d$ of the target set
(and thus also in terms of the input size), 
its worst-case running time is doubly exponential. 
Further, as this algorithm only considers mixing graphs of depth at most $d$, 
it does not always finds an optimal solution (see an example in~\cite{gonzalez2019towards}). 
In spite of these deficiencies, for very small values of $d$ it is still likely to produce 
good mixing graphs. 

Additional work regarding the design of mixing graphs for multiple droplets includes
Huang~{\etal}'s algorithm called $\WARA$, which is an extension of Algorithm~$\REMIA$,
that focuses on reactant minimization; see \cite{huang2013reactant}.
Mitra~{\etal}~\cite{mitra2012chip} also proposed an algorithm for multiple droplet concentrations
by modeling the problem as an instance of the Asymmetric TSP on a de Bruijn graph.

As discussed in~\cite{gonzalez2019towards}, the computational complexity of
computing mixing graphs with minimum waste is still open, even in the case of
single-droplet targets. In fact, it is not even known whether the minimum-waste function is
computable at all, or whether it is decidable to determine if a given target set can
be produced without \emph{any} waste. 
To our knowledge, the only known result that addresses theoretical aspects of
designing mixing graphs is a polynomial-time algorithm in~\cite{gonzalez2019towards} that determines whether
a given collection of droplets with specified concentrations can be mixed perfectly with a mixing graph.

%%%%%%%%%%%%%%%%%%

\myparagraph{Our results}
Continuing the line of work in~\cite{thies2008abstraction,roy2010optimization,huang2012reactant,chiang2013graph},
we develop a new efficient algorithm $\RPRIS$ (for \emph{Recursive Precision Reduction with Initial Shift})
for designing mixing graphs for single-droplet targets, with the objective to minimize waste. 
Our algorithm was designed to provide improved worst-case waste estimate; specifically to cut it by
half for most concentrations. Its main idea is quite natural: recursively, 
at each step it reduces the precision of the target droplet by $2$, while only adding one
waste droplet when adjusting the mixing graph during backtracking.        

While designed with worst-case performance in mind, $\RPRIS$ significantly outperforms algorithms
$\MinMix$, $\DMRW$ and $\GORMA$ in our experimental study,  producing on average about $50\%$ less waste than $\MinMix$,
between $21$ and $25\%$ less waste than $\DMRW$ (with the percentage increasing with 
the precision $d$ of the target droplet), and about $17\%$ less waste than $\GORMA$.
(It also produces about $40\%$ less waste than $\REMIA$.)    
Additionally, when compared to $\ILP$, $\RPRIS$ produces on average only about $7\%$
additional waste.

Unlike earlier work in this area, that was strictly experimental,
we introduce a performance measure for waste minimization algorithms and show that
$\RPRIS$ has better worst-case performance than $\MinMix$ and $\DMRW$. 
This measure is based on two attributes $d$ and $\gamma$ of the target concentration $t$.
As defined earlier, $d$ is the precision of $t$, 
and $\gamma$ is defined as the number of equal leading bits in $t$'s
binary representation, not including the least-significant bit $1$.
For example, if   $t =.00001011$ then  $\gamma = 4$, and if
$t = .1111$ then $\gamma = 3$. (Both $d$ and $\gamma$ are functions of $t$,
but we skip the argument $t$, as it is always understood from context.) 
In the discussion below we provide more intuition and motivations for using
these parameters.

We show that Algorithm $\RPRIS$ produces at most $\half(d+\gamma)+2$ droplets of waste
(see Theorem~\ref{thm: waste bound} in Section~\ref{subsec: performance bounds}).
In comparison, Algorithm~$\MinMix$ from~\cite{thies2008abstraction} produces exactly $d$ droplets of waste
to produce $t$, independently of the value of $t$.
This means that the waste of  $\RPRIS$ is about half that of~$\MinMix$ for almost all concentrations $t$.
(More formally, for a uniformly chosen random $t$ with precision $d$
the probability that the waste is larger than $(\half-\epsilon)d$ vanishes when $d$ grows, for any $\epsilon>0$.) 
As for Algorithm~$\DMRW$, its average waste is better than that of $\MinMix$, but its
worst-case bound is still $d-O(1)$ even 
for small values of $\gamma$ (say, when $t\in [\onefourth, \threefourths]$),
while Algorithm~$\RPRIS$' waste is at most $d/2+O(1)$ in this range.

In regard to time performance, for the problem of computing mixing graphs it would be reasonable
to express the time complexity of an algorithm as a function of its output, which is the
size of the produced graph. This is because the output size is at least as large
as the input size, which is equal to $d$ -- the number of bits of $t$. (In fact, typically it's much larger.)
Algorithm~$\RPRIS$ runs in time that is linear in
the size of the computed graph, and the graphs computed by Algorithm~$\RPRIS$ have size $O(d^2)$.

%%%%%%%%%%%%%%%%%%

\myparagraph{Discussion}       
To understand better our performance measure for waste, observe
that the optimum waste is never smaller than $\gamma+1$. This is because if
the binary representation of $t$ starts with $\gamma$ $0$'s then any
mixing graph has to use $\gamma+1$ input droplets $0$ and at least one droplet $1$.
(The case when the leading bits of $t$ are $1$'s is symmetric.)
For this reasons, a natural approach is to express the waste in the form
$\gamma + f(d-\gamma)$, for some function $f()$. 
In Algorithm~$\RPRIS$ we have $f(x) \approx \half x$.
It is not known whether smaller functions $f()$ can be achieved.

Ideally, one would like to develop efficient ``approximation'' algorithms for waste
minimization, that measure waste performance in terms
of the additive or multiplicative approximation error, with respect to the optimum
value. This is not realistic, however, given the current state of knowledge, since
currently no close and computable bounds for the optimum waste are known.

%%%%%%%%%%%%%%%%%%%%%%%%%%%%%%%%%%%%%%%%%%%%%%%%%%%%%%%%%%%%%%%%%%%%%%%%%%%%%%
%%%%%%%%%%%%%%%%%%%%%%%%%%%%%%%%%%%%%%%%%%%%%%%%%%%%%%%%%%%%%%%%%%%%%%%%%%%%%% 

\section{Preliminaries}
\label{sec: preliminaries}
%%%%%%%%%%%%%%%%%%%%%%%%%%%%%%%%%%%%%%%%%%%%%%%%%%%%%%%%%%%%%%%%%%%%%%%%%%%%%%
%%%%%%%%%%%%%%%%%%%%%%%%%%%%%%%%%%%%%%%%%%%%%%%%%%%%%%%%%%%%%%%%%%%%%%%%%%%%%% 

%\section{Preliminaries}
%\label{sec: preliminaries}
%\input{2_preliminaries.tex}

%%%%%%%%%%%%%%%%%%%%%%%%%%%%%%%%%%%%%%%%%%%%%%%%%%%%%%%%%%%%%%%%%%%%%%%%%%%%%%
%%%%%%%%%%%%%%%%%%%%%%%%%%%%%%%%%%%%%%%%%%%%%%%%%%%%%%%%%%%%%%%%%%%%%%%%%%%%%%

We use notation $\precision(c)$ for the precision of concentration $c$,
that is the number of fractional bits in the binary representation of $c$.
(All concentration values will have finite binary representation.)
In other words, $\precision(c)= d\in\nonnegintegers$ such that
$c = a/2^d$ for an odd $a\in\integers$. 

We will deal with sets of droplets, some possibly with equal concentrations.
We define a \emph{configuration} as a multiset of droplet concentrations.
Let $A$ be an arbitrary configuration.
By $\nofdroplets{A}=n$ we denote the number of droplets in $A$.
We will often write a configuration as
$A = \braced{f_1:a_1,f_2: a_2,...,f_m:a_m}$, where each $a_i$ represents
a different concentration and $f_i$ denotes the multiplicity of $a_i$ in $A$. 
(If $f_i=1$, then, we will just write ``$a_i$'' instead of ``$f_i:a_i$''.)
Naturally, we have $\sum_{i=1}^m f_i = n$. 

We defined mixing graphs in the introduction. A mixing graph can be thought of,
abstractly, as a linear mapping from the source values (usually $0$'s and $1$'s) to the sink
values. Yet in the paper, for convenience, we will assume that the source
concentration vector
is part of a mixing graph's specification, and that all sources, micro-mixers,
and sinks are labeled by their associated concentration values.

We now define an operation of graph coupling. Consider two mixing graphs $G_1$ and $G_2$. 
Let $T_1$ be the output configuration (the concentration labels of the sink nodes) of $G_1$ and
$I_2$ be the input configuration (the concentration labels of the source nodes) for $G_2$.
To construct the \emph{coupling} of $G_1$ and $G_2$, denoted $G_2\bullet G_1$,
we identify inlet edges of the sinks of $G_1$ with labels from $T_1\cap I_2$ with
outlet edges of the corresponding sources in $G_2$. 
More precisely, repeat the following steps as long as $T_1\cap I_2 \neq\emptyset$: 
(1) choose any $a\in T_1\cap I_2$, 
(2) choose any sink node $t_1$ of $G_1$ labeled $a$, and let $(u_1,t_1)$ be its inlet edge,
(3) choose any source node $s_2$ of $G_2$ labeled $a$, and let $(s_2,v_2)$ be its outlet edge,
(4) remove $t_1$ and $s_2$ and their incident edges, and finally,
(5) add edge $(u_1,v_2)$.
The remaining sources of $G_1$ and $G_2$ become sources of $G_2\bullet G_1$, and
the remaining sinks of $G_1$ and $G_2$ become sinks of $G_2\bullet G_1$.
See Figure~\ref{fig: mixing graph composition example} for an example.

\begin{figure}[ht]
	\begin{center}
		\includegraphics[width = 3.4in]{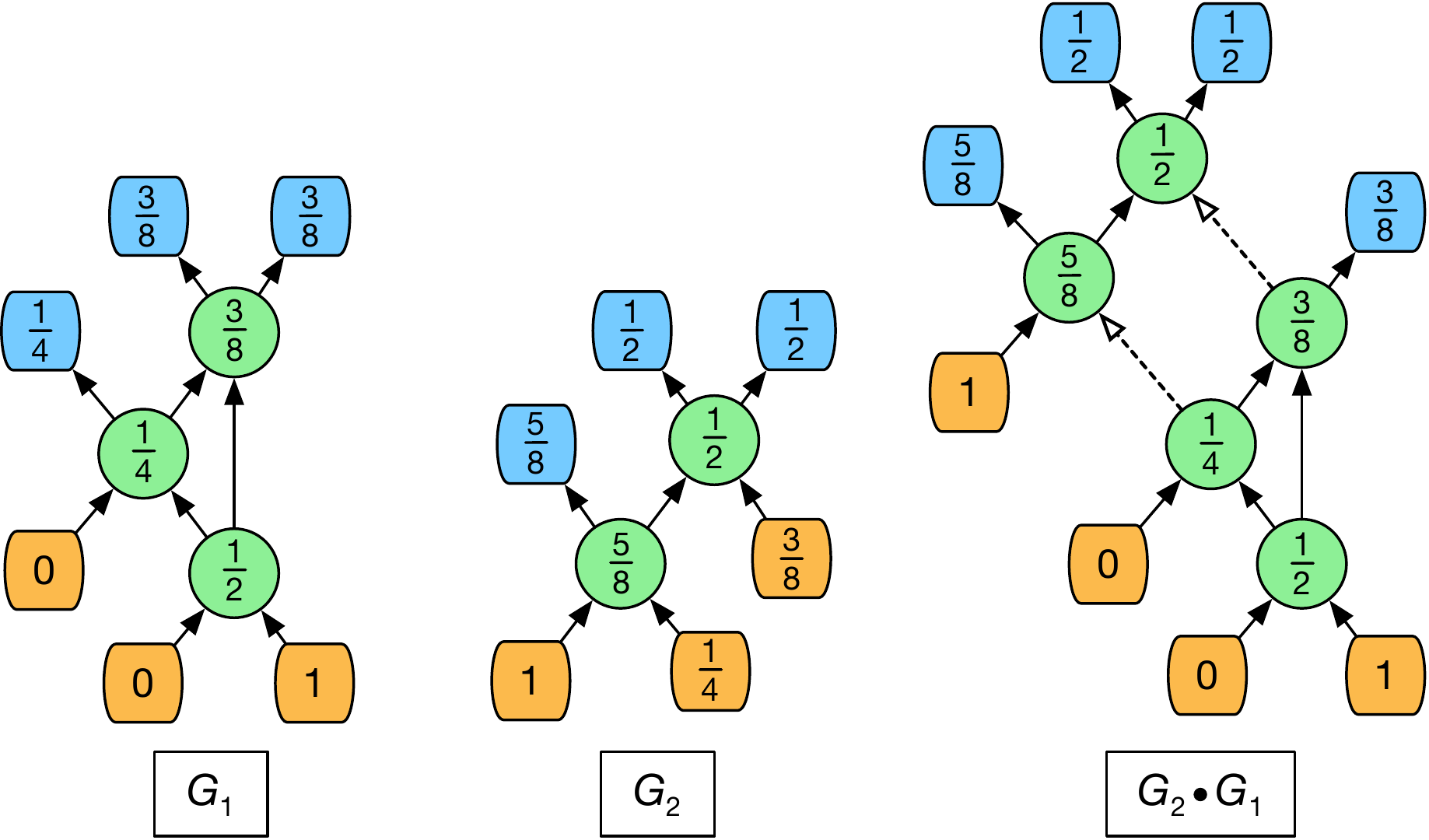}
		\caption{Coupling of two mixing graphs $G_1$ and $G_2$. $G_2\bullet G_1$ is obtained by
		identifying inlet edges of two sinks of $G_1$, one labelled $\onefourth$ and one $\threeeighths$,
		with the outlet edges of the corresponding sources of $G_2$.
		These new edges are shown as dotted arrows.}
		\label{fig: mixing graph composition example}
	\end{center}
\end{figure}

Next, we define converter graphs.
An \emph{$(i:\alpha,j:\beta)$-converter} is a mixing graph
that produces a configuration of the form $T=\braced{i:\alpha,j:\beta}\cup W$,
where $W$ denotes a set of waste droplets,
and whose input droplets have concentration labels either $0$ or $1$.
As an example, graph $G_2$ in Figure~\ref{fig: mixing graph example} can
be interpreted as a $(2:\oneeighth,2:\thirteensixteenths)$-converter that produces two waste droplets
of concentrations $\onefourth$ and $\seveneighths$.

If needed, to avoid clutter, sometimes we will use a more compact graphical representation of mixing graphs 
by aggregating (not necessarily all) nodes with the same concentration labels into a single node, 
and with edges labeled by the number of droplets that flow through them.
(We will never aggregate two micro-mixer nodes if they both produce a droplet of waste.)
If the label of an edge is $1$, then we will simply omit the label.
See Figure~\ref{fig: mixing graph simplification} for an example of such a compact representation.

\begin{figure}[ht]
	\begin{center}
		\includegraphics[width = 3.2in]{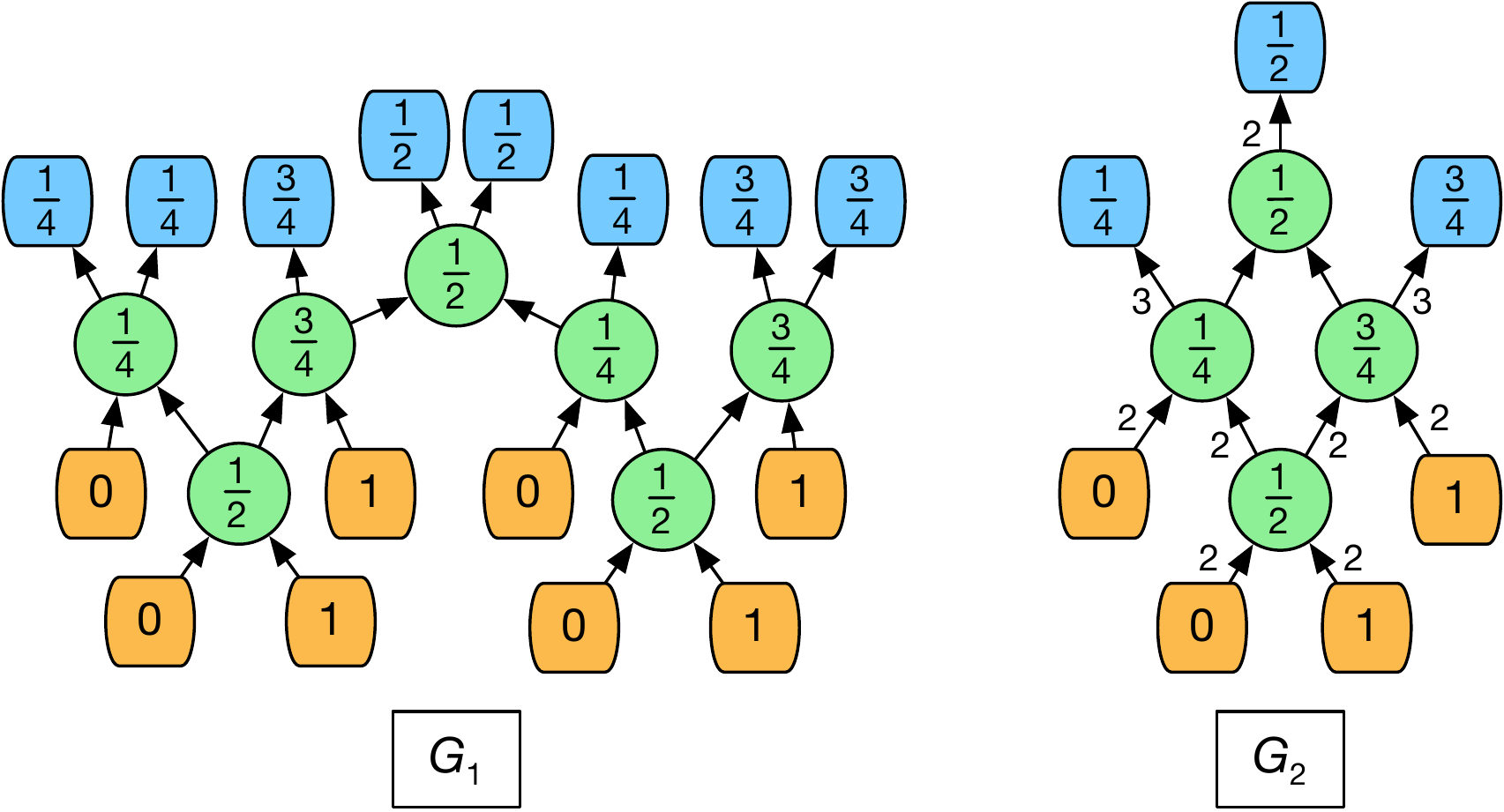}
		\caption{$G_2$ is a compact representation of $G_1$. 
			All nodes in $G_2$ (except the last intermediate node with label $\half$) represent 
			an aggregation of at least two nodes from $G_1$.}
			\label{fig: mixing graph simplification}
		\end{center}
	\end{figure}

%%%%%%%%%%%%%%%%%%%%%%%%%%%%%%%%%%%%%%%%%%%%%%%%%%%%%%%%%%%%%%%%%%%%%%%%%%%%%%
%%%%%%%%%%%%%%%%%%%%%%%%%%%%%%%%%%%%%%%%%%%%%%%%%%%%%%%%%%%%%%%%%%%%%%%%%%%%%% 

\section{Algorithm Description}
\label{sec: algorithm description}
%%%%%%%%%%%%%%%%%%%%%%%%%%%%%%%%%%%%%%%%%%%%%%%%%%%%%%%%%%%%%%%%%%%%%%%%%%%%%%
%%%%%%%%%%%%%%%%%%%%%%%%%%%%%%%%%%%%%%%%%%%%%%%%%%%%%%%%%%%%%%%%%%%%%%%%%%%%%% 

%\section{Algorithm Description}
%\label{sec: algorithm description}
%\input{3_algorithm_description.tex}

%%%%%%%%%%%%%%%%%%%%%%%%%%%%%%%%%%%%%%%%%%%%%%%%%%%%%%%%%%%%%%%%%%%%%%%%%%%%%%
%%%%%%%%%%%%%%%%%%%%%%%%%%%%%%%%%%%%%%%%%%%%%%%%%%%%%%%%%%%%%%%%%%%%%%%%%%%%%%

In this section, we describe our algorithm $\RPRIS$ for producing a single-droplet target
of concentration $t$ with precision $d=\precision(t)$.
We first give the overall strategy and then we gradually explain its implementation.
The core idea behind $\RPRIS$ is a recursive procedure that we refer to as
\emph{Recursive Precision Reduction}, that we outline first. 
In this procedure, $t_s$ denotes the concentration computed at the $s^{th}$ recursive step with $d_s=\precision(t_s)$;
initially, $t_0 = t$. Also, by $\calB$ we denote the set of base concentration values with small
precision for which we give explicit mixing graphs later in this section.
\begin{description}
	\setlength{\itemsep}{0.01in}
	\item{\textbf{Procedure} RPR$(t_s)$}
	
		\textbf{If} $t_s\in \calB$, let $G_s$ be the base mixing graph (defined later) for $t_s$, \textbf{else}:   
		\begin{description}
			\item{(rpr1)} Replace $t_s$ by another concentration value $t_{s+1}$ with $d_{s+1} = d_s-2$.
			\item{(rpr2)} Recursively construct a mixing graph $G_{s+1}$ for $t_{s+1}$.
			\item{(rpr3)} Convert $G_{s+1}$ into a mixing graph $G_s$ for $t_s$, 
				increasing waste by one droplet. 
		\end{description}  
		\textbf{Return} $G_s$.
\end{description}	
The mixing graph produced by this process is $G_0$.

When we convert $G_{s+1}$ into $G_s$ in part~(rpr3), the precision of the target increases by $2$, 
but the waste only increases by $1$, which gives us a rough bound of $d/2$ on the overall waste.
However, the above process does not work for all concentration values; 
it only works when $t_0 \in [\onefourth , \threefourths]$. 
To deal with values outside this interval, we map $t$ into $t_0$
so that $t_0 \in [\onefourth , \threefourths]$,
next we apply Recursive Precision Reduction to $t_0$, and then we appropriately modify
the computed mixing graph. This process is called \emph{Initial Shift}.

We next describe these two processes in more detail, starting with 
Recursive Precision Reduction, followed by Initial Shift.

%%%%%%%%

\myparagraph{Recursive Precision Reduction (RPR)}  
We start with concentration $t_0$ that, by applying Initial Shift (described next), 
we can assume to be in $[\onefourth , \threefourths]$.

%%%%%%%%

\emparagraph{Step (rpr1): computing $t_{s+1}$}
We convert $t_s$ into a carefully chosen concentration $t_{s+1}$ for which  $d_{s+1} = d_s - 2$.
One key idea is to maintain an invariant so that at each recursive step, this new
concentration value $t_{s+1}$ satisfies $t_{s+1}\in [\onefourth,\threefourths]$.  To accomplish
this, we consider five intervals $S_1=[\oneeighth,\threeeighths]$, 
$S_2=[\onefourth,\onehalf]$, $S_3=[\threeeighths,\fiveeighths]$,
$S_4=[\half,\threefourths]$, and $S_5=[\fiveeighths,\seveneighths]$.
We choose an interval $S_k$ that contains $t_s$ ``in the middle'', that is
$S_k=[l,r]$ for $k$ such that $t_s\in [l+\onesixteenth,r-\onesixteenth]$. 
(See  Figure~\ref{fig: intervals graphical representation}.)
We then compute $t_{s+1}=4(t_s-l)$. Note that $t_{s+1}$ satisfies both 
$t_{s+1}\in [\onefourth,\threefourths]$ (that is, our invariant) and $d_{s+1} = d_s - 2$.

\begin{figure}[ht]
	\begin{center}
		\includegraphics[width = 3.2in]{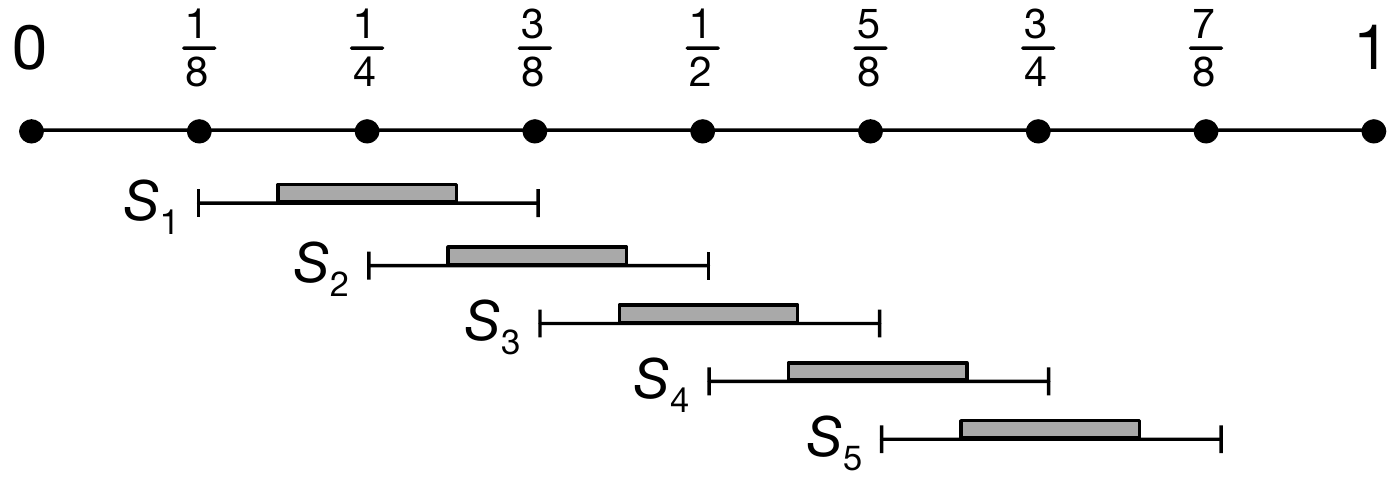}
		\caption{Graphical representation of intervals $S_1,S_2,\dots,S_5$.
		The thick shaded part of each interval $S_k = [l,r]$ marks its ``middle section'' $[l+\onesixteenth,r-\onesixteenth]$.
		Each concentration within interval $[\onefourth,\threefourths]$ belongs to the middle section of some $S_k$.}
		\label{fig: intervals graphical representation}
	\end{center}
\end{figure}
%

%%%%%%%%

\emparagraph{Step (rpr3): converting $G_{s+1}$ into $G_s$}
Let $G_{s+1}$ be the mixing graph obtained for $t_{s+1}$ in step (rpr2). 
We first modify $G_{s+1}$ to obtain a graph $G_{s+1}'$, which is then coupled with an appropriate converter 
$C_{s+1}$ to obtain mixing graph $G_s = G_{s+1}'\bullet C_{s+1}$.
Figure~\ref{fig: recursive mapping} illustrates this process.

\begin{figure}[ht]
	\begin{center}
		\includegraphics[width = 4.2in]{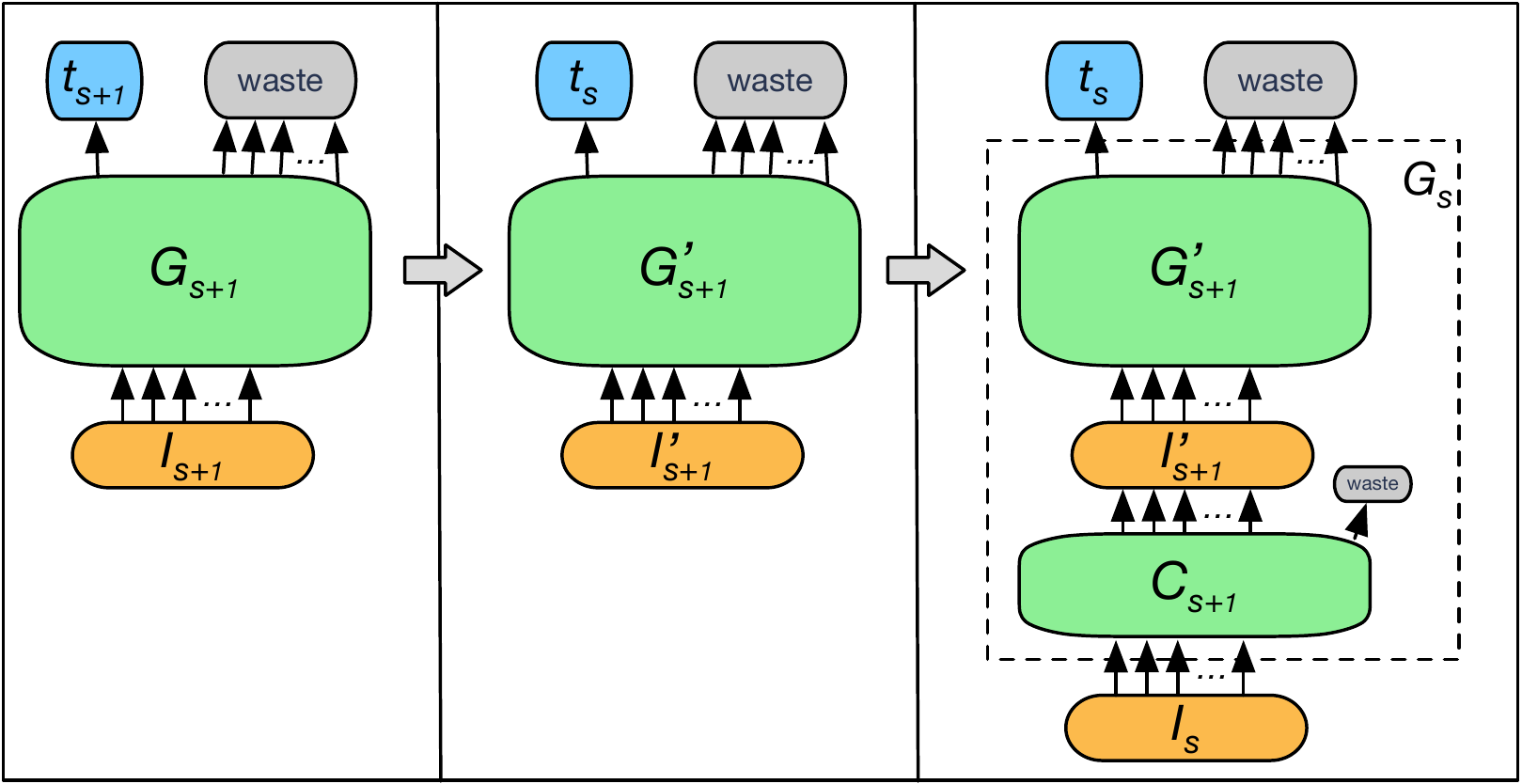}
		\caption{Conversion from $G_{s+1}$ to $G_s$. The left image illustrates the computed
			mixing graph $G_{s+1}$ with input labels $I_{s+1}$ (consisting of only $0$'s and $1$'s)
			 that produces $t_{s+1}$ along with some waste.
			The middle figure illustrates $G_{s+1}'$, which is obtained from $G_{s+1}$
			by changing concentration labels.
			The last figure illustrates the complete mixing graph $G_s = G_{s+1}'\bullet C_{s+1}$ for $t_s$, 
			shown within a dotted rectangle.}
		\label{fig: recursive mapping}
	\end{center}
\end{figure}

Next, we explain how to construct $G_{s+1}'$.
$G_{s+1}'$ consists of the same nodes and edges as $G_{s+1}$, only the concentration
labels are changed.
Specifically, every concentration label $c$ from $G_{s+1}$ is changed to $l+c/4$ in $G_{s+1}'$. 
Note that this is simply the inverse of the linear function that maps $t_s$ to $t_{s+1}$.
In particular, this will map the $0$- and $1$-labels of the
source nodes in $G_{s+1}$ to the endpoints $l$ and $r$ of the corresponding interval $S_k$.

The converter $C_{s+1}$ used in $G_s$ needs to have sink nodes with labels equal to the source nodes for $G_{s+1}'$. 
That is, if the labeling of the source nodes of  $G_{s+1}'$ is
$I_{s+1}'=\braced{i:l,j:r}$, then $C_{s+1}$ will be an $(i:l,j:r)$-converter. 
As a general rule, $C_{s+1}$ should produce at most one waste droplet, but
there will be some exceptional cases where it produces two.
(Nonetheless, we will show that at most one such ``bad'' converter is used during the {RPR} process.)
The construction of these converters is somewhat intricate, and is deferred to the next section.

%%%%%%%%

\myparagraph{The base case}
We now specify the set of base concentration values and their mixing graphs.
Let $\calB=\braced{\half,\onefourth,\threefourths,\threeeighths,\fiveeighths,\fivesixteenths,\elevensixteenths}$.
(Concentrations $\fivesixteenths$ and $\elevensixteenths$ are not strictly necessary
for correctness but are included in the base case to improve the waste bound.)
Figure~\ref{fig: mixing graphs for the base cases} illustrates the mixing graphs for 
concentrations $\half$, $\onefourth$, $\threeeighths$, and $\fivesixteenths$;
the mixing graphs for the remaining concentrations are symmetric.

\begin{figure}[ht]
	\begin{center}
		\includegraphics[width = 3.5in]{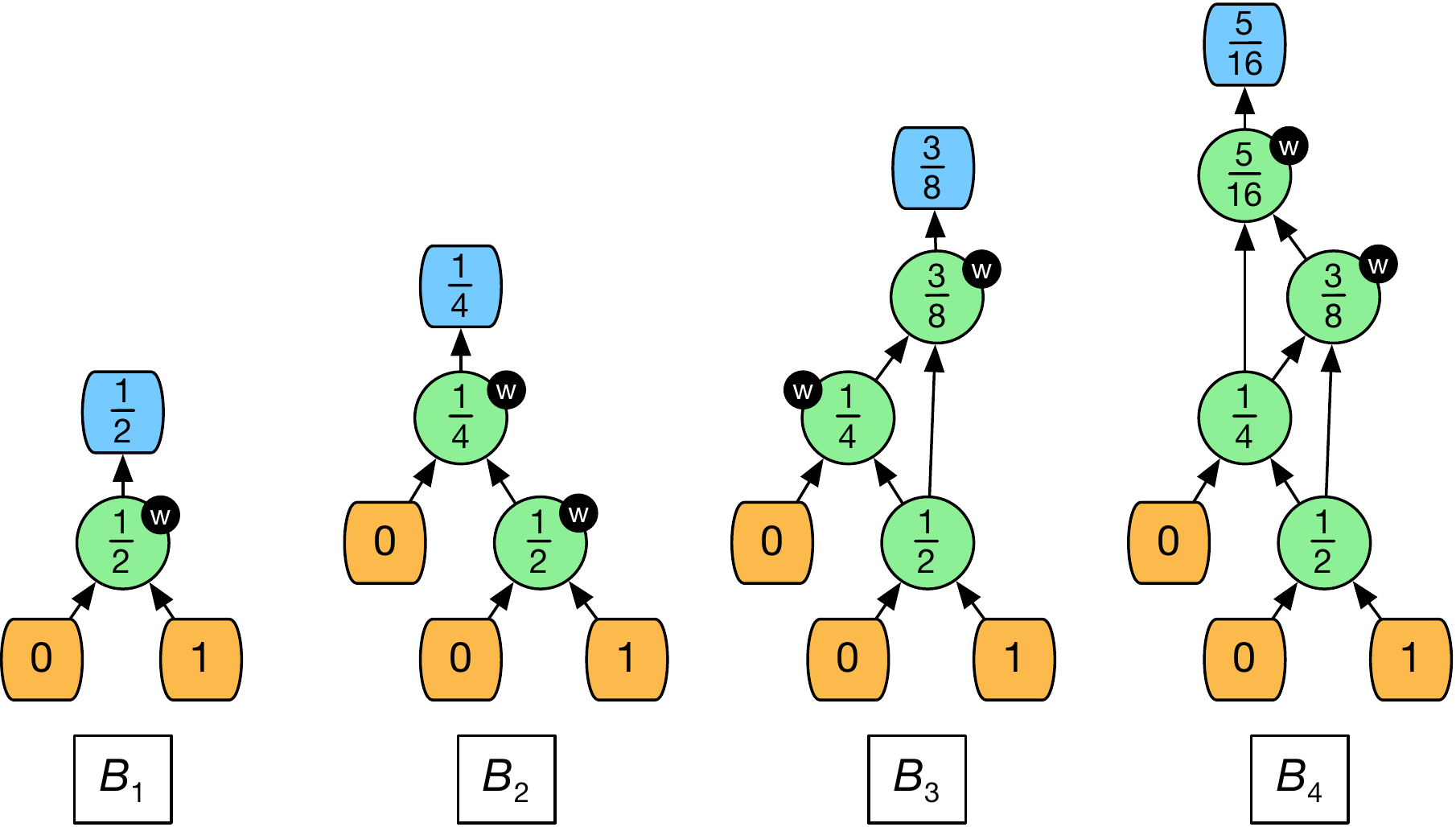}
		\caption{Base mixing graphs $B_1,B_2,B_3$ and $B_4$ for concentrations 
			$\half,\onefourth,\threeeighths$ and $\fivesixteenths$, respectively.}
		\label{fig: mixing graphs for the base cases}
	\end{center}
\end{figure}
%

%%%%%%%%

\myparagraph{Initial Shift (IS)}
We now describe the IS procedure. At the fundamental level, the idea is similar
to a single step of {RPR}, although the involved linear mappings and the converter
are significantly different.

We can assume that $t < \onefourth$ (because for $t > \threefourths$ the process is symmetric).
Thus the binary representation of $t$ starts with $\gamma\geq 2$ fractional $0$'s.
Since $2^{\gamma-1}t\in [\onefourth,\half)$, we could use this value as the result of
the initial shift, but to improve the waste bound we refine this choice as follows:
If $2^{\gamma-1}t\in (\threeeighths,\half)$ then let $t_0 =  2^{\gamma-1}t$ and
$\sigma = 1$. Otherwise, we have $2^{\gamma-1}t\in [\onefourth,\threeeighths]$, in which
case we let $t_0 =  2^{\gamma}t$ and $\sigma = 0$.  In either case,
$t_0 = 2^{\gamma-\sigma}t\in[\onefourth,\threefourths]$ and $d_0 = d - \gamma+\sigma$.

Let $G_0$ be the mixing graph obtained by applying the {RPR} process to $t_0$.
It remains to show how to modify $G_0$ to obtain the mixing graph $G$ for $t$.     
This is analogous to the process shown in  Figure~\ref{fig: recursive mapping}. 
We first construct a mixing graph $G_0'$ that consists of the same nodes and edges as $G_0$, 
only each concentration label $c$ is replaced by $c/2^{\gamma-\sigma}$.
In particular, the label set of the source nodes in $G_0'$ will have the form
$I_0' = \braced{i:0, j: 1/2^{\gamma-\sigma}}$.
We then construct a  $(i:0, j: 1/2^{\gamma-\sigma})$-converter $C_0$ 
and couple it with $G_0'$ to obtain $G$; that is,  $G=G_0'\bullet C_0$.   
This $C_0$ is easy to construct: The $0$'s don't require any mixing, and
to produce the $j$ droplets $1/2^{\gamma-\sigma}$ we start
with one droplet $1$ and repeatedly mix it with $0$'s, making sure to generate at most
one waste droplet at each step.
More specifically, after $z$ steps we will have $j_z$ droplets with
concentration $1/2^z$, where $j_z = \ceiling{j/2^{\gamma-\sigma-z}}$.
In step $z$, mix these $j_z$ droplets with $j_z$ $0$'s, producing
$2j_z$ droplets with concentration $1/2^{z+1}$.   
We then either have $j_{z+1} = 2j_z$, in which case there is no waste,
or  $j_{z+1} = 2j_z-1$, in which case one waste droplet $1/2^{z+1}$ is produced.
Overall, $C_0$ produces at most $\gamma-\sigma$ waste droplets.

%%%%%%%%%%%%%%%%%%%%%%%%%%%%%%%%%%%%%%%%%%%%%%%%%%%%%%%%%%%%%%%%%%%%%%%%%%%%%%
%%%%%%%%%%%%%%%%%%%%%%%%%%%%%%%%%%%%%%%%%%%%%%%%%%%%%%%%%%%%%%%%%%%%%%%%%%%%%% 

\section{Construction of Converters}
\label{sec: construction_of_converters}
%%%%%%%%%%%%%%%%%%%%%%%%%%%%%%%%%%%%%%%%%%%%%%%%%%%%%%%%%%%%%%%%%%%%%%%%%%%%%%
%%%%%%%%%%%%%%%%%%%%%%%%%%%%%%%%%%%%%%%%%%%%%%%%%%%%%%%%%%%%%%%%%%%%%%%%%%%%%% 

%\section{Construction of Converters}
%\label{sec: construction_of_converters}
%\input{4_construction_of_converters.tex}

%%%%%%%%%%%%%%%%%%%%%%%%%%%%%%%%%%%%%%%%%%%%%%%%%%%%%%%%%%%%%%%%%%%%%%%%%%%%%%
%%%%%%%%%%%%%%%%%%%%%%%%%%%%%%%%%%%%%%%%%%%%%%%%%%%%%%%%%%%%%%%%%%%%%%%%%%%%%%

In this section we detail the construction of our converters. 
Let $t_s$ denote the concentration at the $s^{th}$ recursive step in the RPR process.
We can assume that $t_s\in [\onefourth,\half]$, because the case $t_s\in (\half,\threefourths]$
is symmetric.
Recall that for a $t_s$ in this range, in Step~(rpr1) we will chose an appropriate interval
$S_k$, for some $k\in \braced{1,2,3}$.
Let $S_k = [l,r]$ (that is, $ l = k\cdot\oneeighth$ and $r = l+\onefourth$). 
For each such $k$ and all $i,j\ge 1$
we give a construction of an $(i:l,j:r)$-converter that we will denote $C_{i,j}^k$.
Our main objective here is to design these converters so that they 
produce as little waste as possible --- ideally none.

%%%%%%%%%%%%%%%%%%%%%%%%%%%%%%%%%%%%%%%%%%%%%%%%%%%%%%%%%%%%%%%%%%%%%%%%%%%%%%

\subsection{$(i:\onefourth, j:\onehalf)$-Converters $C_{i,j}^2$}
\label{subsec: converters for one fourths and halves}

%%%%%%%%%%%%%%%%%%%%%%%%%%%%%%%%%%%%%%%%%%%%%%%%%%%%%%%%%%%%
% \subsection{Strategy for configurations $A=\braced{a:\onefourth, b:\onehalf}$}
% \label{subsec: one fourths and halves}
% \input{1-1_one-fourth_one-half.tex}
%%%%%%%%%%%%%%%%%%%%%%%%%%%%%%%%%%%%%%%%%%%%%%%%%%%%%%%%%%%%

We start with the case $k=2$, because in this case the construction is relatively simple.
We show how to construct, for all $i,j\ge 1$, our $(i:\onefourth,j:\onehalf)$-converter $C_{i,j}^2$
that produces at most one droplet of waste.   
These converters are constructed via an iterative process. We first give initial converters
$C_{i,j}^2$, for some small values of $i$ and $j$, by providing specific graphs. All other converters
are obtained from these initial converters by repeatedly coupling them with other
mixing graphs that we refer to as \emph{extenders}.  

Let $\initJ^2 =\braced{(i,j)}_{i,j \in\braced{ 1,2}}$.
The initial converters $C_{i,j}^2$ are defined for the four index pairs $(i,j) \in \initJ^2$.
Figure~\ref{fig: one-fourth_one-half} illustrates the initial converters 
$C_{2,1}^2,C_{1,2}^2$ and two extenders $X_1^2,X_2^2$.
Converter $C_{1,2}^2$ produces one waste droplet and  converter $C_{2,1}^2$ does not produce any waste.
Converter $C_{1,1}^2$ can be obtained from $C^2_{2,1}$ by designating one of the $\onefourth$ outputs as waste.
Converter $C_{2,2}^2$ is defined as $C_{2,2}^2=X_1^2\bullet C_{2,1}^2$, and produces one waste droplet of $\half$.
(Thus $C_{2,2}^2$ is simply a disjoint union of $C_{2,1}^2$ and $X_1^2$ with one output $\half$ designated as waste.)

\begin{figure}[ht]
	\begin{center} 
		\includegraphics[width = 3.8in]{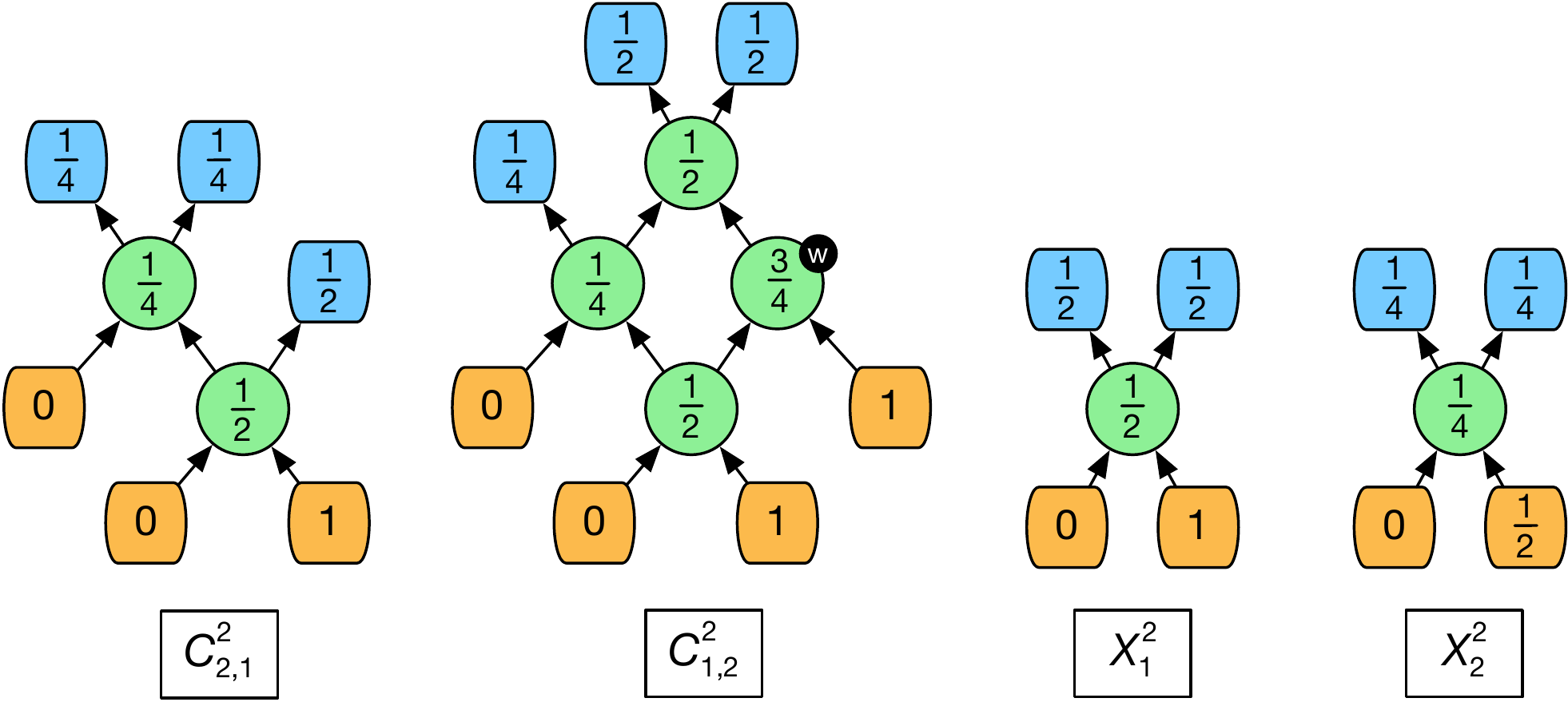}
		\caption{Initial converters and extenders for the case $I=\braced{i:\onefourth,j:\onehalf}$.}
		\label{fig: one-fourth_one-half}
	\end{center}
\end{figure}

The construction of other converters $C_{i,j}^2$ is based on the following
observation: Suppose that we already have constructed some $C_{i,j}^2$.
Then 
(i) $X_1^2\bullet C_{i,j}^2$ is a $C_{i,j + 2}^2$ converter that
produces the same waste as $C_{i,j}^2$, and
(ii) provided that $j\ge 2$,
$X_2^2\bullet C_{i,j}^2$ is a $C_{i+2, j-1}^2$ converter
that produces the same waste as $C_{i,j}^2$.

Let now $i,j\ge 1$ with $(i,j)\notin\initJ^2$ be arbitrary. 
To construct $C_{i,j}^2$, using the initial converters
and the above observation, express the integer vector
$(i,j)$ as $(i,j) = (i',j') + \phi(0,2) + \psi(2,-1)$, 
for some $i',j'\in \initJ^2$ and integers 
$\psi=\lceil \frac{i}{2}\rceil-1$ and $\phi=\lceil \frac{j+\psi}{2}\rceil-1$.
Then $C^2_{i,j}$ is constructed by starting with $C^2_{i',j'}$ and 
coupling it $\phi$ times with $X_1^2$ and then $\psi$ times with $X_2^2$. 
(This order of coupling is not unique but is also not arbitrary, 
because each extender $X_2^2$ requires a droplet of concentration $\half$ as input.)
Since $X_1^2$ and $X_2^2$ do not produce waste, $C^2_{i,j}$ will produce at most
one waste droplet.

%%%%%%%%%%%%%%%%%%%%%%%%%%%%%%%%%%%%%%%%%%%%%%%%%%%%%%%%%%%%%%%%%%%%%%%%%%%%%%

\subsection{$(i:\threeeighths, j:\fiveeighths)$-Converters $C_{i,j}^3$}
\label{subsec: converters for three eighths and five eighths}

%%%%%%%%%%%%%%%%%%%%%%%%%%%%%%%%%%%%%%%%%%%%%%%%%%%%%%%%%%%%
% \subsection{Strategy for configurations $A=\braced{a:\threeeighths, b:\fiveeighths}$ with $a,b\geq 1$}
% \label{subsec: three eighths and five eighths}
% \input{1-3_three-eighths_five-eighths.tex}
%%%%%%%%%%%%%%%%%%%%%%%%%%%%%%%%%%%%%%%%%%%%%%%%%%%%%%%%%%%%

Next, for each pair $i,j\ge 1$ we construct an $(i:\threeeighths,j:\fiveeighths)$-converter $C_{i,j}^3$. 
These converters are designed to produce one droplet of waste. ($C^3_{1,1}$ will be an exception, see the
discussion below).
Our approach follows the scheme from Section~\ref{subsec: converters for one fourths and halves}:
we start with some initial converters, which then can be repeatedly coupled with appropriate extenders
to produce all other converters.
Since concentrations $\threeeighths$ and $\fiveeighths$ are symmetric (as $\fiveeighths = 1 - \threeeighths$),
we will only show the construction of converters $C_{i,j}^3$ for $i \ge j$;
the remaining converters can be computed using symmetric mixing graphs.

\begin{figure}[ht]
	\begin{center} 
		\includegraphics[width = 5.55in]{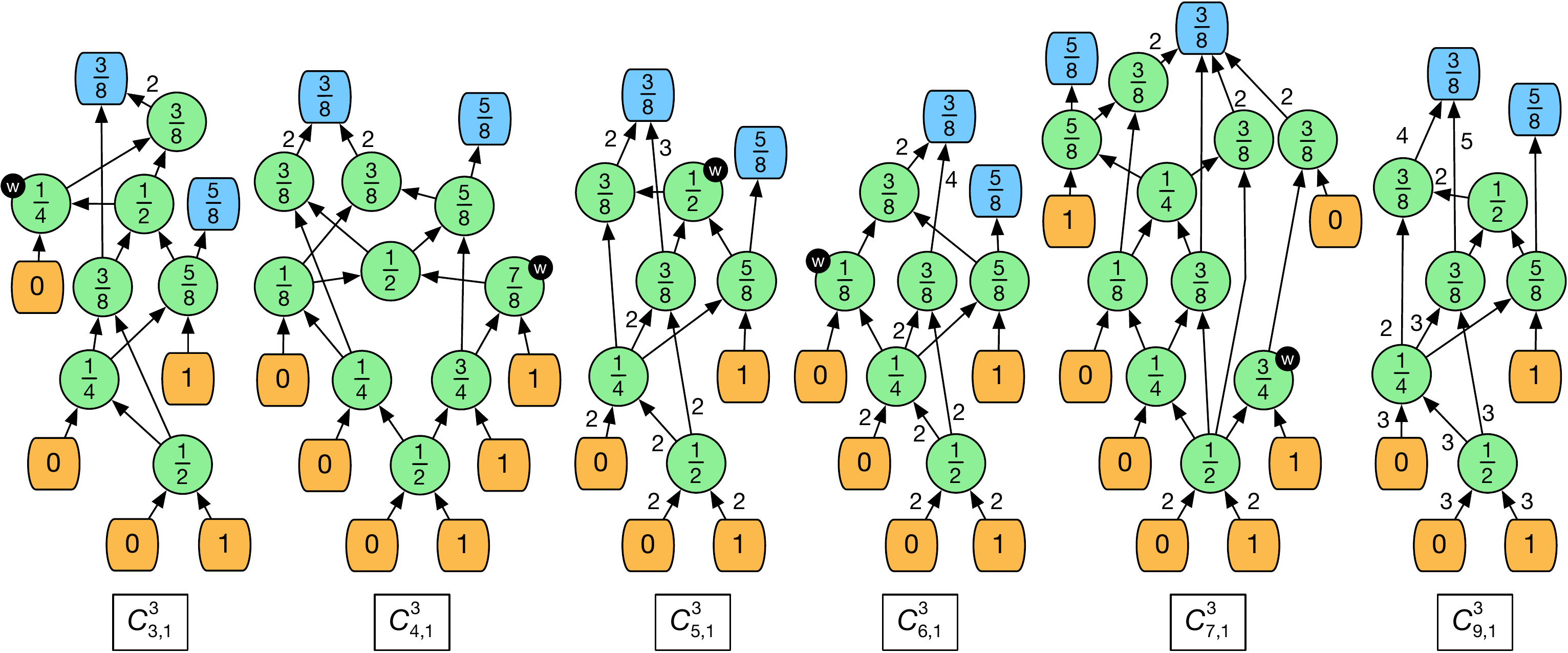}
		\caption{Initial converters for the case $I=\braced{i:\threeeighths,j:\fiveeighths}$.}
		\label{fig: three-eighths five-eighths}
	\end{center}
\end{figure}

Let $\initJ^3 =\braced{(i,1)}_{i\in\braced{1,2,...,9}}\cup \braced{(2,2)}$.
The initial converters $C_{i,j}^2$ are defined for all index pairs $(i,j) \in \initJ^3$.
Figure~\ref{fig: three-eighths five-eighths} shows converters $C_{3,1}^3,C_{4,1}^3,\dots,C_{7,1}^3$ and $C_{9,1}^3$.
Converter $C_{8,1}^1$ can be obtained from $C^3_{9,1}$ by designating an output of $\threeeighths$ as waste.
Converter $C^3_{2,2}$ is almost identical to $X_1^3$ in Figure~\ref{fig: three-eighths five-eighths extenders};
except that the source labels $\threeeighths$ and $\fiveeighths$ are replaced by $0$ and $1$, respectively
(the result of mixing is still $\half$, so other concentrations in the graph are not affected).
Converters $C_{1,1}^3$ and $C_{2,1}^3$ are obtained from $C^3_{2,2}$ by designating outputs of
$\braced{\threeeighths,\fiveeighths}$ and $\fiveeighths$, respectively, as waste.
Note that all initial converters except for $C^3_{1,1}$ produce at most one droplet of waste.

Now, consider extenders $X_1^3$ and $X_2^3$ in Figure~\ref{fig: three-eighths five-eighths extenders}.
The construction of other converters $C_{i,j}^3$ follows the next observation:
Assume that we have already constructed some $C_{i,j}^3$, with $i\ge j$.
Then 
(i) $X_1^3\bullet C_{i,j}^3$ is a $C_{i+1, j+1}^3$ converter that
produces the same waste as $C_{i,j}^3$, and
(ii) $X_2^3\bullet C_{i,j}^3$ is a $C_{i+8,j}^3$ converter
that produces the same waste as $C_{i,j}^3$.
\begin{figure}[ht]
	\begin{center} 
		\includegraphics[width = 2in]{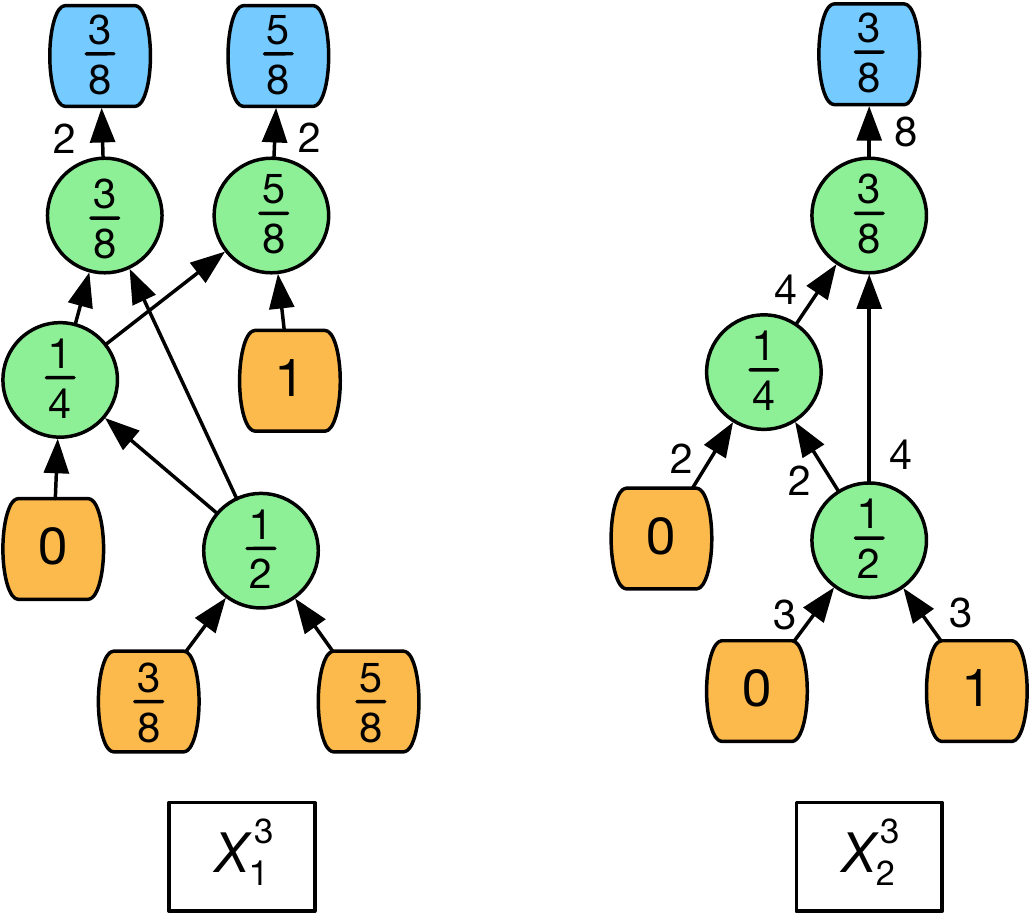}
		\caption{$X_1^3$ and $X_2^3$ extenders for the case $I=\braced{i:\threeeighths,j:\fiveeighths}$.}
		\label{fig: three-eighths five-eighths extenders}
	\end{center}
\end{figure}

Consider now arbitrary $i \ge j\ge 1$ with $(i,j)\notin\initJ^3$. To construct $C_{i,j}^3$, using the initial converters
and the above observation, express the integer vector
$(i,j)$ as $(i,j) = (i',j') + \phi(1,1) + \psi(8,0)$, 
for some integers $\phi,\psi\ge 0$, and $(i',j')\in \initJ^3-\braced{(1,1)}$.
Then $C^3_{i,j}$ is constructed by starting with $C^3_{i',j'}$ and 
coupling it  $\phi$ times with $X_1^3$ and then $\psi$ times with $X_2^3$
(in arbitrary order).
Since $X_1^3$ and $X_2^3$ do not produce waste (and we do not use the initial
converter $C^3_{1,1}$), $C^3_{i,j}$ will produce at most one waste droplet.

Overall, all converters $C^3_{i,j}$, except for $C^3_{1,1}$ produce at most one
waste droplet. Converter $C^3_{1,1}$ produces two droplets of waste; however,
as we later show in Section~\ref{subsec: performance bounds},
it is not actually used in the algorithm.

%%%%%%%%%%%%%%%%%%%%%%%%%%%%%%%%%%%%%%%%%%%%%%%%%%%%%%%%%%%%%%%%%%%%%%%%%%%%%%

\subsection{$(i:\oneeighth, j:\threeeighths)$-Converters $C_{i,j}^1$ }
\label{subsec: converters for one eighths and three eighths}

%%%%%%%%%%%%%%%%%%%%%%%%%%%%%%%%%%%%%%%%%%%%%%%%%%%%%%%%%%%%
% \subsection{Strategy for configurations $A=\braced{a:\oneeighth, b:\threeeighths}$}
% \label{subsec: one eighths and three eighths}
% \input{1-2_one-eighth_three-eighth.tex}
%%%%%%%%%%%%%%%%%%%%%%%%%%%%%%%%%%%%%%%%%%%%%%%%%%%%%%%%%%%%

In this section, 
for each pair $i,j\ge 1$ we construct an $(i:\oneeighth,j:\threeeighths)$-converter $C_{i,j}^1$.
Most of these converters produce at most one droplet of waste, but there will be four exceptional
coverters with waste two. (See the comments at the end of this section.)
The idea of the construction follows the same scheme as in Sections~\ref{subsec: converters for one fourths and halves}
and~\ref{subsec: converters for three eighths and five eighths}: we
start with some initial converters and repeatedly couple them with appropriate extenders to obtain other converters.

\begin{figure}[ht]
	\begin{center} 
		\includegraphics[width = 5.55in]{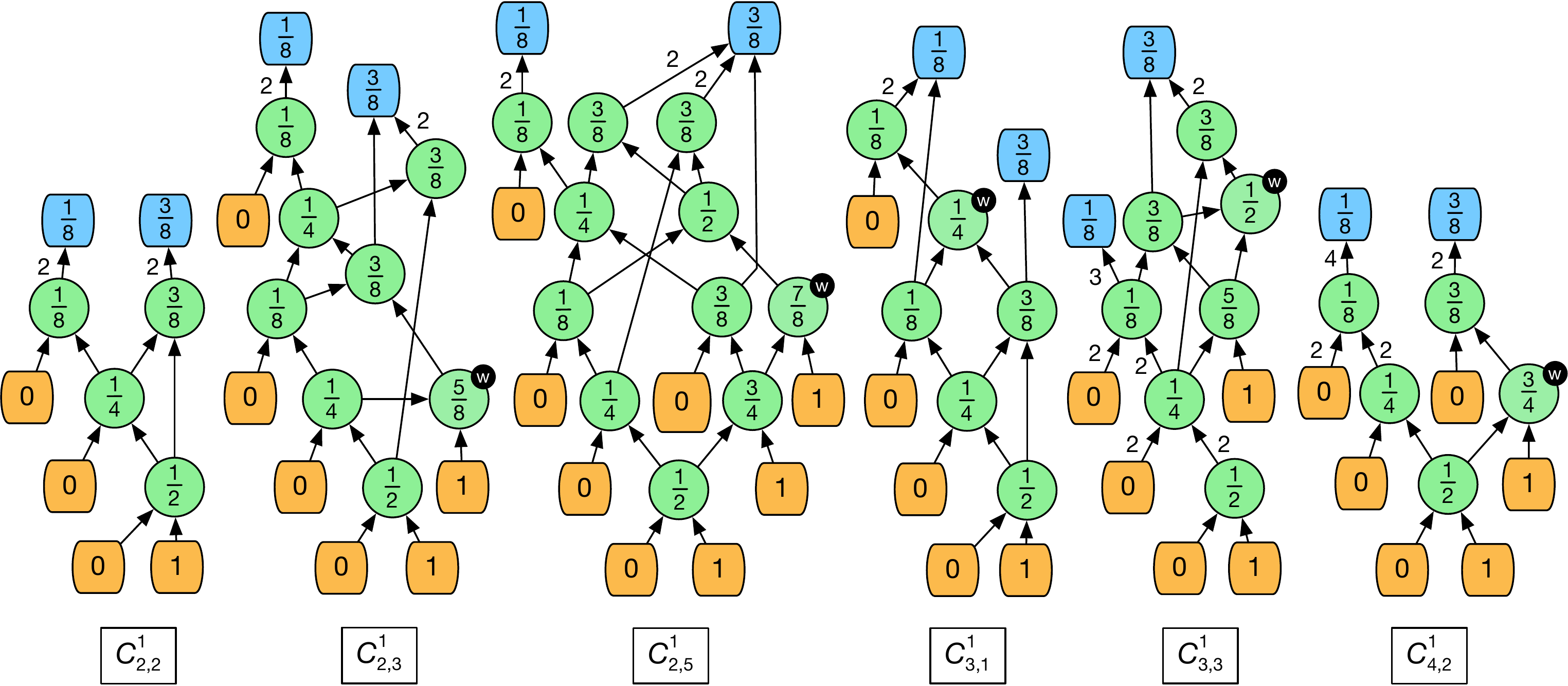}
		\caption{Initial converters for the case $I=\braced{i:\oneeighth,j:\threeeighths}$.}
		\label{fig: one-eighth_three-eighths}
	\end{center}
\end{figure}
Let $\initJ^1 =\braced{(i,j)}_{i,j\in\braced{1,2,3}} \cup \braced{(4,2),(2,5)}$.
The initial converters $C_{i,j}^1$ are defined for all index pairs $(i,j) \in \initJ^1$.
Converters $C_{2,2}^1$, $C_{2,3}^1$, $C_{2,5}^1$, $C_{3,1}^1$, $C_{3,3}^1$ and $C_{4,2}^1$
are shown in Figure~\ref{fig: one-eighth_three-eighths}.
Converters $C_{1,1}^1$, $C_{1,2}^1$ and $C_{2,1}^1$ are obtained from $C_{2,2}^1$ by designating outputs of
$\braced{\oneeighth,\threeeighths}$, $\oneeighth$ and $\threeeighths$, respectively, as waste.
Converter $C_{1,3}^1$ is obtained from $C_{2,3}^1$ by designating an output of $\oneeighth$ as waste, and
$C_{3,2}^1$ is obtained from $C_{4,2}^1$ by designating an output of $\oneeighth$ as waste. 
Thus, among the initial converters,
$C_{1,1}^1$, $C_{1,3}^1$ and $C_{3,2}^1$ each produces two droplets of waste; all other converters
have at most one droplet of waste.

Next, we provide an observation leading to the construction of other converters $C_{i,j}^1$.
Consider extenders $X_1^1$ and $X_2^1$ in Figure~\ref{fig: one-eighth_three-eighths_extenders} and
assume that we have already constructed some $C_{i,j}^1$.
Then, 
(i) provided that $j\ge 2$, $X_1^1\bullet C_{i,j}^1$ is a $C_{i+3, j-1}^1$ converter that
produces the same waste as $C_{i,j}^1$, and
(ii) provided that $i\ge 2$, $X_2^1\bullet C_{i,j}^1$ is a $C_{i-1,j+3}^1$ converter
that produces the same waste as $C_{i,j}^1$.
We also need the following, less obvious observation:

%%%

\begin{observation}\label{cla: pairs (i,j) for k=2}
If $i,j\ge 1$ and $(i,j) \notin \initJ^1\cup\braced{(6,1)}$, 
then $(i,j) = (i',j') + \phi(-1,3) + \psi(3,-1)$,
for some integers $\phi,\psi\ge 0$, and $(i',j')\in \initJ^1-\braced{(1,1),(1,3),(3,2)}$.
\end{observation}

\begin{proof}
Let $i,j\ge 1$ and $(i,j) \notin \initJ^1\cup\braced{(6,1)}$.
We note first that we can represent $(i,j)$ as
$(i,j) = (\tildei,\tildej) + \tildephi(-1,3) + \tildepsi(3,-1)$,
for $(\tildei,\tildej) \in \initJ^1-\braced{(2,5),(4,2)}$ and
integers $\tildephi,\tildepsi \geq 0$.
If $(\tildei,\tildej) \notin \braced{(1,1),(1,3),(3,2)}$ then we are done.
Otherwise, we show how to modify the values of parameters $\tildei$, $\tildej$,
$\tildephi$ and $\tildepsi$ so that they satisfy the condition in the observation.

\smallskip\noindent
\mycase{1} {$(\tildei,\tildej) = (1,1)$}. 
		For this case, $\tildephi,\tildepsi\geq 1$ must hold, as otherwise we would
		get a contradiction with  $i,j\geq 1$.
		Therefore, we can write $(i,j)$ as
		$(i,j) = (3,3) + (\tildephi-1)(-1,3) + (\tildepsi-1)(3,-1)$.

\smallskip\noindent		
\mycase{2} {$(\tildei,\tildej) = (1,3)$}.
		For this case, $\tildepsi\geq 1$ must hold, because $i \geq 1$.
		Therefore, we can write $(i,j)$ as
		$(i,j) = (4,2) + \tildephi(-1,3) + (\tildepsi-1)(3,-1)$.	

\smallskip\noindent		
\mycase{3}{$(\tildei,\tildej) = (3,2)$}.
		For this case, it is sufficient to prove that $\tildephi\geq 1$, 
		since we could then write $(i,j)$
		as $(i,j) = (2,5) + (\tildephi-1)(-1,3) + \tildepsi(3,-1)$.
		To show that $\tildephi\geq 1$ we argue by contradiction, as follows.
		Suppose that $\tildephi=0$. 
		Then $(i,j) = (3,2) + \tildepsi(3,-1)$.
		For $\tildepsi\in\braced{0,1}$
		this contradicts that $(i,j) \notin \initJ^1\cup\braced{(6,1)}$,
		and for $\tildepsi\ge 2$ it contradicts that $j\ge 1$.
\end{proof}

%%%

%
\begin{figure}[ht]
	\begin{center} 
		\includegraphics[width = 2.3in]{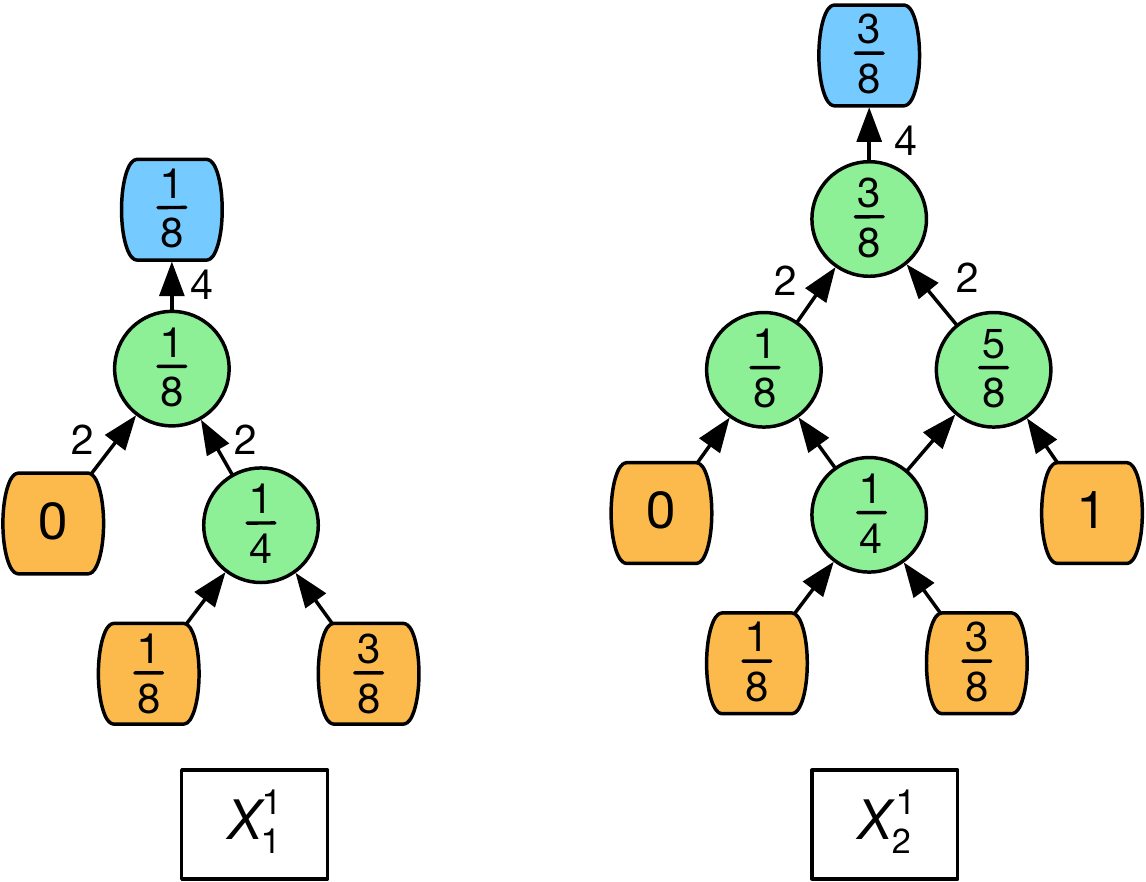}
		\caption{$X_1^1$ and $X_2^1$ extenders for the case $I=\braced{i:\oneeighth,j:\threeeighths}$.}
		\label{fig: one-eighth_three-eighths_extenders}
	\end{center}
\end{figure}

Using the observations above, for any pairs $i,j\ge 1$ we can construct converter $C^1_{i,j}$ as follows.
If $(i,j) = (6,1)$ we let $C^1_{6,1}=X_1^1\bullet C^1_{3,2}$ (so $C^1_{6,1}$ has two droplets of waste).
If $(i,j) \neq (6,1)$, we construct $C^1_{i,j}$ by starting with $C^1_{i',j'}$ and repeatedly 
coupling it with $\phi$ copies of $X_2^1$ and $\psi$ copies of $X_1^1$, choosing a suitable order of
couplings to ensure that each intermediate
converter has at least one output $\oneeighth$ and at least one $\threeeighths$.
(For example, if $j'=1$ then we begin by coupling $X_2^1$ first.)
As $X_1^1$ and $X_2^1$ do not produce any waste, these $C^1_{i,j}$'s will each produce at most one droplet of waste.

Overall, the converters $C^1_{i,j}$ we construct have at most one droplet of waste, with the exception of the
following four: $C^1_{1,1}$, $C^1_{1,3}$, $C^1_{3,2}$ and $C^1_{1,6}$. (It is easy to prove that for these
converters waste $2$ cannot be avoided.)
As we show later in Section~\ref{subsec: performance bounds}, of these four converters only $C^1_{1,3}$ is actually
used in the RPR process of Algorithm~$\RPRIS$, and it is used at most once.

%%%%%%%%%%%%%%%%%%%%%%%%%%%%%%%%%%%%%%%%%%%%%%%%%%%%%%%%%%%%%%%%%%%%%%%%%%%%%%
%%%%%%%%%%%%%%%%%%%%%%%%%%%%%%%%%%%%%%%%%%%%%%%%%%%%%%%%%%%%%%%%%%%%%%%%%%%%%%

\section{Performance Bounds}
\label{subsec: performance bounds}
%%%%%%%%%%%%%%%%%%%%%%%%%%%%%%%%%%%%%%%%%%%%%%%%%%%%%%%%%%%%%%%%%%%%%%%%%%%%%%
%%%%%%%%%%%%%%%%%%%%%%%%%%%%%%%%%%%%%%%%%%%%%%%%%%%%%%%%%%%%%%%%%%%%%%%%%%%%%%

%\section{Performance Bounds}
%\label{subsec: performance bounds}
%\input{5_performance_bounds.tex}   

%%%%%%%%%%

In this section we provide the analysis of Algorithm~$\RPRIS$, including the
worst-case bound on produced waste, a bound on the size of computed mixing graphs, and
the running time.

%%%%%%%%%%

\myparagraph{Bound on waste}
We first estimate the number of waste droplets of Algorithm~$\RPRIS$. 
Let $G$ be the mixing graph constructed by $\RPRIS$ for
a target concentration $t$ with its corresponding values $d=\precision(t)$
and $\gamma$ (as defined in Section~\ref{sec: introduction}).
Below we prove the following theorem.

\begin{theorem}\label{thm: waste bound}
The number of waste droplets in $G$ is at most $\half(d+\gamma)+2$.
\end{theorem}

To prove Theorem~\ref{thm: waste bound}, we show that the total
number of sink nodes in $G$ is at most $\half(d+\gamma-\sigma)+3$,
for corresponding $\sigma\in\braced{0,1}$.
(This is sufficient, as one sink node is used to produce $t$).
 
Following the algorithm description in Section~\ref{sec: algorithm description}, let $G=G_0'\bullet C_0$.
From our construction of $C_0$ (at the end of Section~\ref{sec: algorithm description}),
we get that $C_0$ contributes at most $\gamma-\sigma$ sink nodes to $G$.
(Each waste droplet produced by $C_0$ represents a sink node in $G$.)
Therefore, to prove Theorem~\ref{thm: waste bound} it remains to show 
that $G_0'$ contains at most $\half(d-\gamma+\sigma) + 3$ sink nodes.
This is equivalent to showing that $G_0$, computed by process {RPR} for $t_0$ 
(and used to compute $G_0'$),
contains at most $\half d_0 + 3$ sink nodes, where $d_0 = \precision(t_0) = d - \gamma + \sigma$.
Lemma~\ref{lem: sink nodes in hatG'} next proves this claim.

%%%%%%%%%%

\begin{lemma}\label{lem: sink nodes in hatG'}
The number of sink nodes in $G_0$ is at most $\half d_0 + 3$.
\end{lemma}

%%%%%%%%%%

\begin{proof}
Let $t_b$ be the concentration used for the base case of the RPR process and 
$d_b=\precision(t_b)\leq d_0$ its precision. 
We prove the lemma in three steps.
First, we show that (i) the number of sink nodes in the mixing graph 
computed for $t_b$ is at most three. (In particular, this gives us that the lemma holds if $t_0 = t_b$.)
Then, we show that (ii) 
if $t_0\neq t_b$ then the number of converters used in the construction of $G_0$ is no more than
$\half d_0 - 1$, and (iii) that at most one of such converter contains two waste sink nodes.
All sink nodes of $G_0$ are either in its base-case graph or in its converters, so
combining claims (i), (ii) and~(iii) gives a complete proof for Lemma~\ref{lem: sink nodes in hatG'}.

The proof of (i) is by straightforward inspection. By definition of the base case,
$t_b\in\calB= \braced{\half,\onefourth,\threefourths,\threeeighths,\fiveeighths,\fivesixteenths,\elevensixteenths}$.
The mixing graphs for base concentrations are shown in Figure~\ref{fig: mixing graphs for the base cases}.
(The graphs for $\threefourths$, $\fiveeighths$, and $\elevensixteenths$ are symmetric to $B_2$,
$B_3$, and $B_4$.) All these graphs have at most $3$ sink nodes.

Next, we prove part~(ii). 
In each step of the RPR process we reduce the precision of the target concentration
by $2$ until we reach the base case, which gives us that 
the number of converters is exactly $\half(d_0-d_b)$.
It is thus sufficient to show that  $d_b\ge 2$, as this immediately implies~(ii).
Indeed, the assumption that $t_0\neq t_b$ and the definition of the base
case implies that $d_0\ge 4$. (This is because the algorithm maintains the 
invariant that its target concentration is in $[\onefourth,\threefourths]$ and all
concentrations in this interval with precision at most $3$ are in $\calB$.)
This, and the precision of the target concentration decreasing by exactly $2$ in each step of the recursion, 
imply that $d_b \in\braced{2,3}$ holds.

We now address part~(iii).
First we observe that converters $C_{1,1}^k$ are not used in the construction of $G_0$:
If we did use $C_{1,1}^k$ in the construction of $G_0$ then
the source labels for the next recursive step are $\braced{0,1}$. Hence, $t_b=\half$.
Now, let $t_{b-1}$ be the concentration, and $S_k=[l,r]$ the interval, used to compute $t_b$.
Since $t_b=\half$, then $t_{b-1}=\half(l+r)$. Therefore, by definition of $S_k$, 
$t_{b-1}\in \braced{\onefourth,\threeeighths,\half,\fiveeighths,\threefourths}\subset \calB$,
so Algorithm~{\RPRIS} would actually use a base case mixing graph for $t_{b-1}$, 
instead of constructing $C_{1,1}^k$ for $t_b$.

So, it is sufficient to consider $C_{i,j}^k$ converters that satisfy $i+j\geq 3$ with $i,j\geq 1$.
Now, from Sections~\ref{subsec: converters for one fourths and halves},
\ref{subsec: converters for three eighths and five eighths} 
and~\ref{subsec: converters for one eighths and three eighths},
we observe that the only such converters that contain two waste sink nodes are
$C_{1,3}^1,C_{3,2}^1$ and $C_{6,1}^1$. Claim~\ref{cla: unused converters} below shows that
converters $C_{6,1}^1$ and $C_{3,2}^1$ are not used in the construction of $G_0$.

Regarding $C_{1,3}^1$, first we note that this converter has
exactly six source nodes; see Figure~\ref{fig: one-eighth_three-eighths}, 
Section~\ref{subsec: converters for one eighths and three eighths}.
This implies that $C_{1,3}^1$ can not be used more than once in the construction of $G_0$,
since the number of source nodes at each recursive step in the RPR process is decreasing.
(Note that there are symmetric converters $C_{3,1}^5$, $C_{2,3}^5$ and $C_{1,6}^5$
for $C_{1,3}^1$, $C_{3,2}^1$ and $C_{6,1}^1$, respectively,
where superscript $5$ is associated to interval $S_5$.
Nevertheless, a similar argument holds.) Thus, step~(iii) holds.

\begin{claim}\label{cla: unused converters}
	Converters $C_{6,1}^1$ and $C_{3,2}^1$ are not used by Algorithm~$\RPRIS$
	in the construction of $G_0$ for $t_0$.
\end{claim}

	We first present the following observations.
	Consider recursive step $s$ of the RPR process, for which
	$t_s$ is the target concentration.
	If a converter $C_{i,j}^1$ is used in this step, 
	then $t_s\in(\onefourth,\fivesixteenths]$ must hold; 
	that is $t_s$ is in the middle part of interval $S_1$ 
	(see Figure~\ref{fig: intervals graphical representation} in
	Section~\ref{sec: algorithm description}).
	(Recall that, by our algorithm's invariant, $t_s\in[\onefourth,\threefourths]$.
	Also, note that $t_s\neq\onefourth$ since otherwise this would be a base case
	and the algorithm would use $B_2$ from Figure~\ref{fig: mixing graphs for the base cases} instead.)
	Further, at the next step of the RPR process, 
	$t_{s+1}=4(t_s-\oneeighth)$ satisfies $t_{s+1}\in(\half, \threefourths]$.
	
	We now prove the claim by contradiction, using the above observations.
	Assume that either $C_{6,1}^1$ or $C_{3,2}^1$ were used in the construction of $G_0$.
	If $C_{6,1}^1$ was used in the construction of $G_0$, 
	then the concentration labels of the source nodes at the next recursive step are $\braced{6:0,1}$, 
	and thus, since $t_{s+1}>\half$, there is not enough reactant available to produce $t_{s+1}$.

	On the other hand, if $C_{3,2}^1$ was used in the construction of $G_0$, then
	the concentration labels of the source nodes at the next 
	recursive step are $\braced{3:0,2:1}$.
	This implies that the next step is guaranteed not to be a base case,
	since all mixing graphs used for base case concentrations contain at most three source nodes,
	as illustrated in Figure~\ref{fig: mixing graphs for the base cases}.
	Now, as $t_{s+1}>\half$, depending on the exact value of $t_{s+1}$,
	the chosen interval for $t_{s+1}$ must be either 
	$S_3=[\threeeighths,\fiveeighths]$,	$S_4=[\half,\threefourths]$ or $S_5=[\fiveeighths,\seveneighths]$.
	We now consider these three cases.

\smallskip
\noindent
\mycase{1} $t_{s+1} \in (\half,\ninesixteenths]$. Then the chosen interval is
		$S_3=[\threeeighths,\fiveeighths]$.
		The only $C_{i,j}^3$ converter with source concentration labels $\braced{3:0,2:1}$ 
		is $C_{3,1}^3$ (see in Figure~\ref{fig: three-eighths five-eighths}
		in Section~\ref{subsec: converters for three eighths and five eighths}),
		whose sink nodes have concentration labels $\braced{\onefourth,3:\threeeighths,\fiveeighths}$.
		Therefore, the input configuration for the next recursive step will be a subset of $\braced{3:0,1}$, 
		which does not have enough reactant to produce $4(t_{s+1}-\threeeighths)>\half$,
		thus contradicting the choice of $S_3$.
		
\smallskip
\noindent
\mycase{2} $t_{s+1} \in (\ninesixteenths,\elevensixteenths]$. Then the chosen interval is   
		$S_4=[\half,\threefourths]$.
		This instance is symmetric to interval $S_2$, having source concentration labels 
		$\braced{2:0,3:1}$, instead of $\braced{3:0,2:1}$, and target concentration $t_{s+1}'=(1-t_{s+1})$.
		Thus we proceed accordingly.
		Since every converter and extender in Section~\ref{subsec: converters for one fourths and halves}
		adds at least the same number of source nodes with concentration label $0$
		as source nodes with concentration label $1$,
		then no converter constructed by the algorithm will have
		source concentration labels $\braced{2:0,3:1}$.
		Hence, we have a contradiction with
		the choice of $S_2$ for $t_{s+1}'$, and thus also with the choice of $S_4$ for $t_{s+1}$.

\smallskip
\noindent		
\mycase{3} $t_{s+1} \in (\elevensixteenths,\threefourths]$. Then the chosen interval is 
		$S_5=[\fiveeighths,\seveneighths]$.
		The argument here is simple: to produce concentration $\seveneighths$,
		at least three reactant droplets are needed, but the input configuration contains only two.
		Therefore, at the next recursive step, the algorithm will not have enough reactant droplets 
		to construct a converter $C^5_{i,j}$ with $i,j\geq 1$,
		contradicting the choice of $S_5$ for $t_{s+1}$.

\medskip

	Finally, neither $S_3,S_4$ nor $S_5$ are chosen by our algorithm for $t_{s+1}$,
	contradicting $C_{3,2}^1$ being used for the construction of $G_0$.
	
	This completes the proof of Claim~\ref{cla: unused converters}
	and Lemma~\ref{lem: sink nodes in hatG'} (thus also completing the proof of
	Theorem~\ref{thm: waste bound}).
\end{proof}

%%%%%%%%%%%%%%%%%%%%

\myparagraph{Size of mixing graphs and running time}
Let $G=G_0'\bullet C_0$ be the mixing graph computed by Algorithm~$\RPRIS$ for $t$;
$C_0$ is constructed by process IS while
$G_0'$ is obtained from $G_0$ (constructed by process RPR) by changing concentration labels appropriately.
We claim that the running time of Algorithm~$\RPRIS$ is $O(|G|)$,
and that the size of $G$ is $O(d^2)$, for $d=\precision(t)$.
We give bounds for $G_0$ and $C_0$ individually, then we combine them to obtain the claimed bounds.
(This is sufficient because the size of $G_0'$, as well as the running time to construct it, 
is asymptotically the same as that for $G_0$.)

First, following the description of process RPR in Section~\ref{sec: algorithm description},
suppose that at recursive step $s$, $G_{s+1}$, $G_{s+1}'$ and converter $C_{s+1} = C^k_{i,j}$ are computed.
(Note that the algorithm does not need to explicitly relabel $G_{s+1}$ to get 
$G_{s+1}'$ -- we only distinguish $G_{s+1}$ from $G_{s+1}'$ for the purpose of presentation.)
The size of $C^k_{i,j}$ is $O(i+j)$ and it takes time $O(i+j)$ to assemble it (as the
number of required extenders is $O(i+j)$). Coupling $C_{s+1}$ with $G_{s+1}'$ also takes time
$O(i+j)$, since $I_{s+1}'$ (the input configuration for $G_{s+1}'$) has cardinality $O(i+j)$ as well. 
In other words, the running time of each recursive RPR step
is proportional to the number of added nodes. Thus the overall running time to construct $G_0$ is $O(|G_0|)$.

Now, let $t_0$ be the target concentration for the RPR process, with $d_0=\precision(t_0)$.
Then, the size of $G_0$ is $O(d_0^2)$. This is because the depth
of recursion in the RPR process is $O(d_0)$, and each converter used in this process
has size $O(d_0)$ as well. The reason for this bound on the converter size is that,
from a level of recursion to the next, the number of source nodes 
increases by at most one (with an exception of at most one step, as explained earlier in this section),
and the size of a converter $C^k_{i,j}$ used at this level is asymptotically the same as the 
number of source nodes at this level.
($I_s$ and $I_{s+1}'$ in Figure~\ref{fig: recursive mapping} illustrate the idea.)

Regarding the bounds for $C_0$, we first argue that the running time to construct $C_0$ is $O(|C_0|)$.
This follows from the construction given in Section~\ref{sec: algorithm description};
in step $s$ there are $2j_s$ droplets being mixed, which requires $j_s$ nodes;
thus the entire step takes time $O(j_s)$.

We next show that the size of $C_0$ is $O(d^2)$. Let $I_0$ be the input configuration for $G_0$.
From the analysis for $G_0$, we get that $|I_0|=O(d_0)$, so the last step in $C_0$ contains $O(d_0)$ nodes.
Therefore, as the depth of $C_0$ is $\gamma-\sigma$, the size of $C_0$ is $O(\gamma d_0)=O(d_0^2)$.

Combining the bounds from $G_0$ and $C_0$, we get that the
running time of Algorithm~{\RPRIS} is $O(|G|)$ and the size of $G$ is $O(d^2)$. 
(The coupling of $C_0$ with $G_0'$ does not affect the overall running time,
since it takes $O(d_0)$ time to couple them, as $|I_0|=O(d_0)$.)

%%%%%%%%%%%%%%%%%%%%%%%%%%%%%%%%%%%%%%%%%%%%%%%%%%%%%%%%%%%%%%%%%%%%%%%%%%%%%%
%%%%%%%%%%%%%%%%%%%%%%%%%%%%%%%%%%%%%%%%%%%%%%%%%%%%%%%%%%%%%%%%%%%%%%%%%%%%%%

\section{Experimental Study}
\label{subsec: experimental study}
%%%%%%%%%%%%%%%%%%%%%%%%%%%%%%%%%%%%%%%%%%%%%%%%%%%%%%%%%%%%%%%%%%%%%%%%%%%%%%
%%%%%%%%%%%%%%%%%%%%%%%%%%%%%%%%%%%%%%%%%%%%%%%%%%%%%%%%%%%%%%%%%%%%%%%%%%%%%%

%\section{Experimental Study}
%\label{subsec: experimental study}
%\input{6_experimental_study.tex}

In this section we compare the performance of our algorithm
with algorithms $\MinMix$, $\REMIA$, $\DMRW$, $\GORMA$ and $\ILP$.
We start with brief descriptions of these algorithms, to give the reader some
intuitions behind different approaches for constructing mixing graphs.
Let $t\in(0,1)$ be the target concentration and $d=\precision(t)$ its precision.
Also, let $\binrep(t)$ be $t$'s binary representation with no trailing zeros.

\begin{description}
	\item[\MinMix~\cite{thies2008abstraction}:] 
	This algorithm is very simple. It starts with $\tau = 0$ and
	mixes it with the bits of $\binrep(t)$ in reverse order, ending with $\tau = t$.
	It runs in time $O(d)$ and produces $d$ droplets of waste.
	
	\item[\REMIA~\cite{huang2012reactant}:]
	This algorithm is based on two phases.
	In the first phase, the algorithm computes a mixing graph $G'$ whose source nodes
	have concentration labels that have exactly one bit $1$ in their binary representation;
	each such concentration represents each of the $1$ bits in $\binrep(t)$.
	Then, in the second phase, a mixing graph $G''$ (that minimizes reactant usage), 
	whose sink nodes are basically a superset of the source nodes in $G'$, is computed.
	Finally, $G$ for $t$ is obtained as $G'\bullet G''$.  
	(Although  $\REMIA$ targets reactant usage, its comparison to different algorithms
	in terms of total waste was also reported in~\cite{huang2012reactant}.    
	Thus, for the sake of completeness, we included $\REMIA$ in our study.)
	
	\item[\DMRW~\cite{roy2010optimization}:]
	This algorithm is based on binary search. Starting with pivot values $l=0$ and $r=1$,
	the algorithm repeatedly ``mixes'' $l$ and $r$ and resets one of them to their average $\half(l+r)$,
	maintaining the invariant that $t\in [l,r]$. After $d$ steps we end up with $l = r = t$. Then the algorithm 
	gradually backtracks to determine, for each intermediate pivot value, how many times this value was used 
	in mixing, and based on this information it computes the required number of droplets. This information
	is then converted into a mixing graph.

	\item[\GORMA~\cite{chiang2013graph}:]
	This algorithm enumerates the mixing graphs for a given target concentration.
	An initial mixing graph is constructed in a top-down manner;
	starting from the target concentration $t$ (the root node), 
	the algorithm computes two concentrations $x$ and $y$
	(called a preceding pair) such that $t = \half(x+y)$ and both $x$ and $y$ have smaller precision than $t$;
	$x$ and $y$ become $t$'s children and both $x$ and $y$ are then processed recursively.
	(Note that a concentration might have many distinct preceding pairs. Each preceding pair is processed.)
	A droplet sharing process is then applied to every enumerated mixing graph to decrease reactant usage and waste produced.
	A branch-and-bound approach is adopted to ease its exponential running time.

	\item[\ILP~\cite{dinh2014network}:]
	This algorithm constructs a ``universal'' mixing graph that contains all
	mixing graphs of depth $d$ as subgraphs.
	It then formulates the problem of computing a mixing graph minimizing
	waste as an integer linear program (a restricted flow problem), and
	solves this program. This universal graph has size exponential in $d$,
	and thus the overall running time is doubly exponential in $d$.
	
\end{description}

We now present the results of our experiments.
Each experiment consisted on generating all concentration values with precision $d$, 
for $d\in\braced{7,8,15,20}$, and comparing the outputs of each of the algorithms.
The results for $\GORMA$ and $\ILP$ are shown only for $d\in\braced{7,8}$, since for 
$d\in \braced{15,20}$ the running time of both $\GORMA$ and $\ILP$ is prohibitive.

\begin{figure}[ht]
	\begin{center}
		\includegraphics[width = 4.5in]{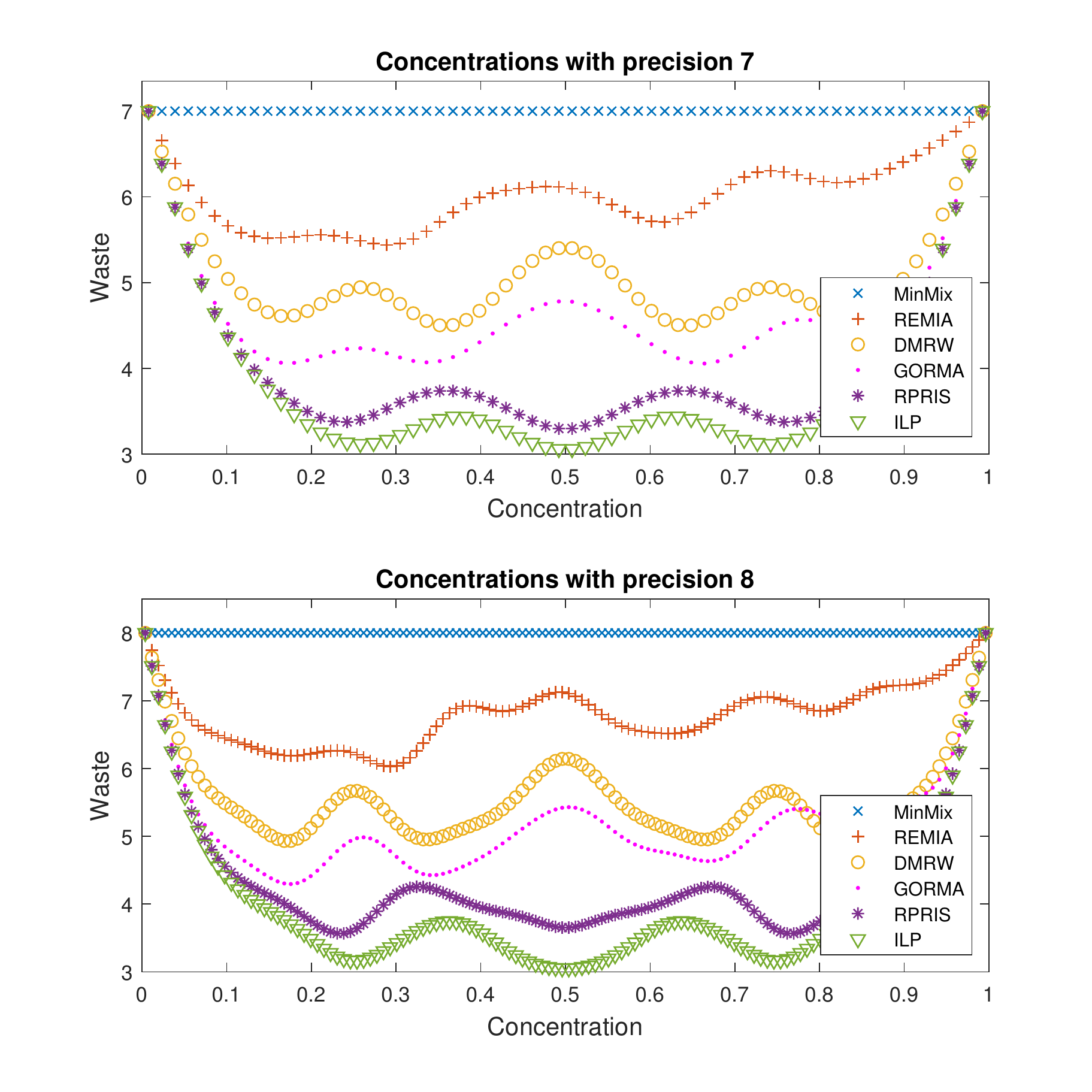}
		\caption{The number of waste droplets of algorithms $\MinMix$, $\REMIA$, 
			$\DMRW$, $\GORMA$, $\ILP$, and our algorithm $\RPRIS$,
		for all concentrations with precision $7$ (top figure)
		and $8$ (bottom figure). All graphs are smoothed using $\MATLAB$'s \emph{smooth} function.
		}
		\label{fig: experiments for precisions 7 and 8.}
	\end{center}
\end{figure}

Figure~\ref{fig: experiments for precisions 7 and 8.} illustrates
the experiments for concentrations of precision $7$ and $8$.
Figure~\ref{fig: experiments for precisions 15 and 20.} illustrates
the experiments for concentrations of precision $15$ and $20$.
In both figures, the data was smoothed using $\MATLAB$'s \emph{smooth} function to reduce
clutter and to bring out the differences in performance between different algorithms.   

As can be seen from these graphs, $\RPRIS$ significantly outperforms algorithm $\MinMix$, $\REMIA$, $\DMRW$ and $\GORMA$:

It produces on average about $50\%$ less waste than $\MinMix$ 
(consistently with our bound of $\half(d+\gamma)+4$ on waste produced by $\RPRIS$),
and $40\%$ less waste than $\REMIA$. It also produces on average
between $21$ and $25\%$ less waste than $\DMRW$, with this percentage increasing with $d$.
Additionally, for $d=7,8$, $\RPRIS$ produces on average about $17\%$ less waste than $\GORMA$ 
and only about $7\%$ additional waste than $\ILP$.

\begin{figure}[ht]
	\begin{center}
		\includegraphics[width = 4.5in]{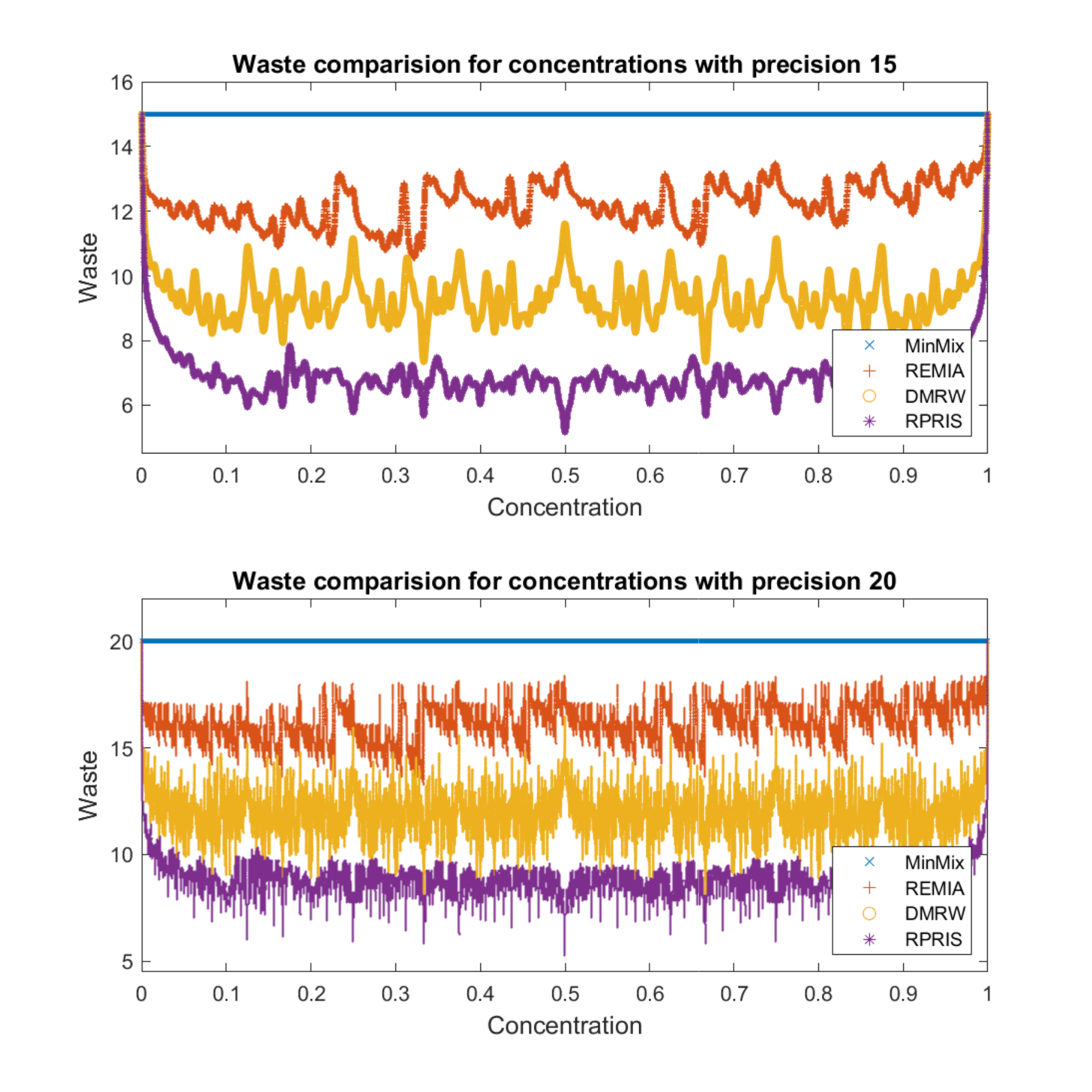}
		\caption{The number of waste droplets of algorithms $\MinMix$,  $\DMRW$,
		          $\REMIA$, and our algorithm $\RPRIS$,
		for all concentrations with precision $15$ (top figure)
		and $20$ (bottom figure). All graphs are smoothed using $\MATLAB$'s \emph{smooth} function.
		}
		\label{fig: experiments for precisions 15 and 20.}
	\end{center}
\end{figure}

Among all of the target concentration values used in our experiments, there is not a single case where $\RPRIS$ 
is worse than either $\MinMix$ or $\REMIA$. 
When compared to $\DMRW$, $\RPRIS$ never produces more waste for precision $7$ and $8$.
For precision $15$, the percentage of concentrations where $\RPRIS$ produces more waste
than $\DMRW$ is below $2\%$, and for precision $20$ it is below $3.5\%$.
Finally, when compared to $\GORMA$, the percentage of concentrations where $\RPRIS$ produces
more waste is below $4\%$.

%%%%%%%%%%%%%%%%%%%%%%%%%%%%%%%%%%%%%%%%%%%%%%%%%%%%%%%%%%%%%%%%%%%%%%%%%%%%%%
%%%%%%%%%%%%%%%%%%%%%%%%%%%%%%%%%%%%%%%%%%%%%%%%%%%%%%%%%%%%%%%%%%%%%%%%%%%%%%

\section{Final Comments}
\label{subsec: final comments}
%%%%%%%%%%%%%%%%%%%%%%%%%%%%%%%%%%%%%%%%%%%%%%%%%%%%%%%%%%%%%%%%%%%%%%%%%%%%%%
%%%%%%%%%%%%%%%%%%%%%%%%%%%%%%%%%%%%%%%%%%%%%%%%%%%%%%%%%%%%%%%%%%%%%%%%%%%%%%
%\section{Final Comments}
%\label{subsec: final comments}
%\input{7_final_comments.tex}
%%%%%%%%%%%%%%%%%%%%%%%%%%%%%%%%%%%%%%%%%%%%%%%%%%%%%%%%%%%%%%%%%%%%%%%%%%%%%%
%%%%%%%%%%%%%%%%%%%%%%%%%%%%%%%%%%%%%%%%%%%%%%%%%%%%%%%%%%%%%%%%%%%%%%%%%%%%%%

In this paper we proposed Algorithm~$\RPRIS$ for single-droplet targets, and
we showed that it outperforms standard waste minimization algorithms $\MinMix$ and $\DMRW$
in experimental comparison. We also proved that its worst-case bound on
waste is also significantly better than for the other two algorithms.

Many questions about mixing graphs remain open. 
We suspect that our bound on waste can be significantly improved.
It is not clear whether waste linear in $d$ is needed for concentrations not too close to $0$ or $1$,
say in $[\onefourth,\threefourths]$.    
In fact, we are not aware of even a \emph{super-constant} (in terms of $d$) lower bound on waste
for concentrations in this range.

For single-droplet targets
it is not known whether minimum-waste mixing graphs can be effectively computed.    
The most fascinating open question, in our view, is whether it is decidable to
determine if a given multiple-droplet target set can be produced without any waste. (As mentioned in 
Section~\ref{sec: introduction}, the ILP-based algorithm from~\cite{dinh2014network} 
does not always produce an optimum solution.)

Another interesting problem is about designing mixing graphs for producing multiple
droplets of the same concentration. Using perfect-mixing graphs from~\cite{gonzalez2019towards},
it is not difficult to prove that if the number of droplets exceeds a certain threshold
then such target sets can be produced with at most one waste droplet.
However, this threshold value is very large and the resulting algorithm very complicated.
As such target sets are of practical significance, a simple algorithm with good
performance would be of interest.     

It would also be interesting to extend our proposed worst-case performance
measure to reactant minimization. It is quite possible that our general approach
of recursive precision reduction could be adapted to this problem.

%%%%%%%%%%%%%%%%%%%%%%%%%%%%%%%%%%%%%%%%%%%%%%%%%%%%%%%%%%%%%%%%%%%%%%%%%%%%%%
%  References %%%%%%%%%%%%%%%%%%%%%%%%%%%%%%%%%%%%%%%%%%%%%%%%%%%%%%%%%%%%%%%%
%%%%%%%%%%%%%%%%%%%%%%%%%%%%%%%%%%%%%%%%%%%%%%%%%%%%%%%%%%%%%%%%%%%%%%%%%%%%%%

\bibliographystyle{plainurl}

\end{document}